%% file: arxiv.tex
\documentclass[11pt]{article}

\usepackage[T1]{fontenc}
\usepackage[utf8]{inputenc}
\usepackage{amsmath,amssymb,amsfonts,amsthm}
\usepackage{bm}
\usepackage{mathtools}
\usepackage{mathpartir}
\usepackage{stmaryrd}
\usepackage{array}
\usepackage{tabularx}
\usepackage{booktabs}
\usepackage{microtype}
\usepackage{ifthen}
\usepackage{xcolor}
\usepackage{tikz}
\usetikzlibrary{arrows.meta,positioning,fit,calc,decorations.pathreplacing,shapes.geometric,shapes,arrows}
\usepackage{algorithm}
\usepackage{algpseudocode}
\usepackage{float}

\usepackage{arxiv}

\usepackage[colorlinks=true,allcolors=blue]{hyperref}
\usepackage{url}

\definecolor{backcolour}{rgb}{0.95,0.95,0.92}

\newtheorem{theorem}{Theorem}[section]
\newtheorem{lemma}[theorem]{Lemma}
\newtheorem{proposition}[theorem]{Proposition}

\theoremstyle{definition}
\newtheorem{definition}{Definition}[section]

\newcommand{\pair}[2]{\langle #1, #2 \rangle}
\newcommand{\encs}[2]{\mathsf{enc}_s(#1,#2)}
\newcommand{\encpk}[2]{\mathsf{enc}_{pk}(#1,#2)}
\newcommand{\hash}[1]{\mathsf{hash}(#1)}
\newcommand{\pk}[1]{\mathsf{pk}(#1)}

\newcommand{\know}{\mathsf{DY}}
\newcommand{\Atoms}{\mathcal{A}}
\newcommand{\Term}{\mathcal{T}}
\newcommand{\synth}{\mathsf{synth}}
\newcommand{\analz}{\mathsf{analz}}
\newcommand{\vdashDY}{\mathrel{\vdash_{\mathsf{DY}}}}
\newcommand{\Roles}{\mathsf{Roles}}
\newcommand{\Ag}{\mathsf{Ag}}

\newcommand{\Sys}{S}
\newcommand{\Init}{\mathsf{Init}}

\newcommand{\Secret}{\mathsf{Secret}}

\newcommand{\hatmap}[1]{\widehat{#1}}
\newcommand{\act}{\alpha}

\newcommand{\todo}[1]{}
\newcommand{\fixme}[1]{}
\newcommand{\expl}[1]{}
\newcommand{\ioana}[1]{}
\newcommand{\fortunat}[1]{}
\newcommand{\srinibas}[1]{}
\newcommand{\jam}[1]{}

\input{example-macros}

\newboolean{long}
\setboolean{long}{true}

\title{Decidability of  Parameterised Dolev-Yao Secrecy}

\author{%
\begin{tabular}{c@{\hspace{2.5em}}c@{\hspace{2.5em}}c}
  Ioana Boureanu & R.\ Ramanujam & Srinibas Swain \\[0.3em]
  {\small Surrey Centre for} & {\small IMSc Chennai} & {\small Adelaide University} \\
  {\small Cyber Security} & {\small Azim Premji University} & \\
  {\small University of Surrey} & & \\
\end{tabular}%
}
\date{}

\reporttitle{From Classical to Parameterised Secrecy}
\reportauthor{Boureanu, Ramanujam \& Swain}

\begin{document}

\maketitle

\begin{abstract}
We study the verification of \textbf{parameterised secrecy} for cryptographic protocols in the Dolev--Yao model, where the number of protocol sessions is unbounded and treated as a parameter. This differs fundamentally from classical Dolev--Yao secrecy, which asks whether a  protocol leaks a secret irrespective of the number of executions; our question is whether secrecy holds uniformly across \emph{all system sizes}, where such a  size is a parameter. This parameterised perspective captures how attacks scale with the number of participants and provides a formal basis for the empirical effectiveness of small-instance analysis.

Secrecy (parameterised or not) is undecidable in general, even under bounded freshness or bounded message size. We identify two structural restrictions that make parameterised secrecy decidable: (i) global bounded freshness per role, and (ii) a Dolev--Yao intruder restricted to well-typed substitutions. Under these assumptions, protocol executions admit a finite representation up to a collapsing map on agents and terms.

Our main result is that parameterised secrecy is decidable in this setting. We obtain a cut-off theorem: secrecy violations in systems with arbitrarily many sessions are always witnessed in systems of bounded size. The cut-off is self-contained; more strongly, the induced transition system forms a well-structured transition system (WSTS) under a bound-based ordering, so secrecy also reduces to a coverability problem in WSTS. This provides a structural explanation for the existence of finite witnesses in symbolic protocol analysis and connects Dolev--Yao verification with parameterised verification techniques.
\end{abstract}

\section{Introduction}
\input{intro2}

\section{Related Work} \label{sec:rw}
\input{relw}

\input{struct}

\section{Preliminaries} \label{prelim}
\subsection{Terms, Signatures, and Dolev--Yao Deduction}
\label{subsec:terms-dy}
\input{dy-stuff}

\subsection{Dolev--Yao Deduction System \& Normal-form DY Proofs}
\label{app:dy}
\input{dy-app}

\subsection{Well-Structured Transition Systems}
\label{subsec:wsts}
\input{wsts-b}

\section{Protocol Model \& Parameterised Secrecy}
\label{sec:protocol-model}

\input{prot-model}

\section{Two  Examples}
\label{sec:examples}
\input{example-freshness}
\input{example-param}

\section{Generic Maps for Agents \& Orders on States} \label{sec:bounds}
\input{boundedness}

\section{Specific Maps \& Orders for Parameterised Secrecy} \label{sec:wqo}
\input{wqo}

\section{Parameterised Secrecy with Bounded Freshness is Decidable} \label{sec:decid}

\input{wsts_result2}

\section{Actually Deciding Parameterised Secrecy} \label{sec:decision}
\input{decid-alg2}

\section{Scope and Limitations} \label{sec:scope}
\input{scope-limitations}

\section{Conclusions} \label{sec:concl}
\input{concl}

\bibliographystyle{unsrt}
\bibliography{biblio1}

\end{document}

%% file: example-macros.tex
\providecommand{\Pex}{\mathcal{P}_{\mathsf{fr}}}
\providecommand{\Pksd}{\mathcal{P}_{\mathsf{ksd}}}

%% file: intro2.tex
\paragraph{\textbf{Context and motivation: Sources of undecidability}}
The formal verification of secrecy properties of cryptographic protocols in the Dolev--Yao (DY) model~\cite{DY83} has been a central concern of the community for several decades. In this symbolic setting, cryptographic primitives are assumed perfect, and the adversary is granted full control over the network: messages can be intercepted, replayed, composed, and decomposed according to algebraic rules. Within this abstraction, secrecy reduces to a reachability question: whether designated data  ever becomes derivable by the intruder.

At a technical level, the undecidability of secrecy stems from two orthogonal sources of unboundedness: unbounded nonces and unbounded message length~\cite{Dur99}.
In the former case, even when symbolic terms remain of bounded size, an adversary can exploit an unbounded supply of fresh values to encode arbitrarily complex information. In the latter case, even with a bounded set of atomic values, the adversary can force honest agents to construct increasingly long messages, thereby simulating unbounded computation.
These phenomena are not merely artefacts of pathological encodings: they are intrinsic to the expressive power of the DY model, as already observed in early undecidability results~\cite{Dur99}.
In particular, bounding only the number of sessions or only the freshness of data does not suffice to recover decidability in general.

\vspace*{0.0cm}
\paragraph{\textbf{Classical routes to decidability}}
A natural approach is therefore to impose suitable bounds. Durgin~\cite{Dur99} showed that secrecy becomes decidable, even under unbounded sessions,  when both the number of nonces and the size of terms are externally bounded. Other approaches impose syntactic discipline instead of explicit bounds. In particular, Ramanujam and Suresh~\cite{RamSu03} established decidability for tagged protocols with unbounded nonces under atomic (well-typed) substitutions, and further showed~\cite{RamSu03b} that tagging alone suffices for decidability without bounds on nonces or freshness.

\vspace*{0.0cm}
\paragraph{\textbf{Attacker model and well-typed substitutions}}
A crucial aspect underlying several of these results, and central to our work, is the role of \emph{well-typed substitutions}, e.g., see \cite{RamSu03b}. In the classical Dolev--Yao setting, the attacker is unrestricted: arbitrary substitutions may be used, enabling so-called type-flaw attacks and contributing directly to undecidability. In contrast, throughout this paper we consider a \emph{restricted Dolev--Yao attacker} whose actions are confined to well-typed substitutions. That is, variables of a given type (e.g., nonces, keys) may only be instantiated by terms of the same type.

This restriction does \emph{not} trivialise the model: the attacker retains full symbolic power in terms of message construction and derivability, but is prevented from exploiting ill-typed encodings. As we will show, this restriction is not merely technical but structural: it is  what enables the collapsing arguments in the parameterised setting and underpins the decidability result.

\vspace*{-0.1cm}
\paragraph{\textbf{From bounds to structure}}
These results suggest a deeper principle: decidability is not solely a matter of bounding resources, but of imposing \emph{structure} on executions. Tagging, typing, and related mechanisms constrain how information flows through runs, effectively reducing the space of behaviours that must be considered.

This raises a \textbf{natural question}: \emph{are there other structuring mechanisms, beyond tagging, that yield decidability?} More specifically, \emph{since decidability arguments ultimately rely on reducing analysis to a finite set of representative runs, can this reduction be achieved via a principled quotienting of executions?}

\vspace*{0.0cm}
\paragraph{\textbf{Our parameterisation}}
In this paper, we pursue this question through the lens of parameterised systems~\cite{PVSurvey2016}.
Parameterisation introduces a form of symmetry and uniformity that can be exploited to collapse large systems into smaller representatives. Here, \emph{we think of protocol executions as being parameterised on roles, or rather agents instantiating one  role or another; this is arguably natural to Dolev-Yao settings. We consider families of systems of arbitrary size (arbitrarily fixed number of agents per role in each family), and study secrecy uniformly across these families}.

Concretely, under bounded freshness and well-typed substitutions, we show that runs in arbitrarily large systems can be quotiented into runs over a bounded number of agents per role, while preserving the intruder's capabilities (modulo the well-typed restriction discussed above). In particular, the restriction to well-typed substitutions ensures that the attacker's message-generation power stabilises: beyond a certain point, newly generated messages do not yield fundamentally new derivations but can be collapsed onto previously generated ones.

This perspective provides a new form of structural restriction: instead of bounding individual resources  term-size directly, we bound the \emph{essential diversity} of runs via parameterised symmetry. The resulting quotienting process yields a finite basis of representative executions, thereby enabling decidability.

\vspace*{0.0cm}
\paragraph{\textbf{Our shift in viewpoint}}
This leads to a conceptual shift. Classical DY-secrecy verification asks: \emph{is secrecy decidable under certain assumptions?} In contrast, the parameterised perspective asks: \emph{how large does a system need to be to witness a secrecy violation?}

This provides a direct explanation for the empirical success of small-model~\cite{LoweF,Cortier} exploration in protocol analysis.

\input{fig-landscape}

\input{bounded-freshness-meaning}

\paragraph{\textbf{Our model: structural restrictions enabling decidability}}
We make the following modelling choices, which are central to our results:

\begin{itemize}
\item \textbf{Global bounded freshness per role.} For each role $X$, a fixed bound $k_X$ limits the \emph{total} number of fresh nonces that may originate across \emph{all} its instances, across \emph{all} runs, rather than per session.

\item \textbf{Canonical reuse via collapse.} When this global budget is exceeded, nonce generation is forced to reuse existing atoms in a controlled, role-consistent manner. This induces a canonical collapsing of agents and terms.

\item \textbf{Well-typed adversary.} The Dolev–Yao intruder is restricted to well-typed substitutions, preventing type-confusion encodings while preserving symbolic derivation power.

\item \textbf{Terminal runs.} Due to the bounds and the well-typedness, every run reaches a terminal point after finitely many steps, beyond which no new atoms are introduced and the intruder's knowledge evolves only by recombination of existing terms.
\end{itemize}

These restrictions, together,  ensure that the adversary's ability to generate new information stabilises, enabling the collapse of large systems into smaller representatives.

\vspace*{0.0cm}
\paragraph{\textbf{Positioning with respect to our decidability results}}
Our approach complements existing decidability results.  Work by Cortier,
Delaune, and collaborators~\cite{CDS21,Del2} establishes small session bounds via
a \emph{semantic} notion of type compliance that restricts the protocol class but
not the attacker. In contrast, our well-typedness assumption is
\emph{operational}: it constrains the attacker (by restricting substitutions)
while preserving a large, practically relevant class of protocols. This is what
enables the structural collapse in parameterised systems. Relatedly, our setting
differs from standard bounded-freshness models~\cite{RamSu03b, TipleaBEB08} in
imposing a \emph{global} per-role freshness budget (rather than per session) and a
well-typed adversary; both are essential to the collapse.

\paragraph{\textbf{Summary of comparisons}}
The conceptual distinction between classical DY secrecy and our parameterised-secrecy approaches here is summarised in Table~\ref{tab:intro}; this  reiterates the explicit role of  restrictions and the shift from semantic to structural reasoning.

\input{table}

\paragraph{\textbf{Contribution}}
In summary, we show that parameterisation provides a principled structuring mechanism for Dolev--Yao protocol analysis. By combining bounded freshness with a well-typed attacker, we obtain a collapsing property on runs, a cut-off theorem for (parameterised) secrecy, and a decidability result that holds uniformly across all system sizes.

%% file: fig-landscape.tex

\begin{figure}[t]
\centering
\small
\renewcommand{\arraystretch}{1.35}
\begin{tabularx}{\linewidth}{@{}l|>{\centering\arraybackslash}X|>{\centering\arraybackslash}X@{}|}
\multicolumn{1}{c}{} &
\multicolumn{1}{c}{\textbf{Sessions bounded}} &
\multicolumn{1}{c}{\textbf{Sessions unbounded}} \\[2pt]
\cline{2-3}
\textbf{Freshness unbounded} &
Decidable {\footnotesize (bounded-session analysis)} &
\textbf{Undecidable} in general {\footnotesize (classical negative results; message size unbounded)} \\
\cline{2-3}
\textbf{Freshness bounded} &
Decidable {\footnotesize (few atoms, few sessions)} &
\textbf{This work:} decidable via parameterised collapse $+$ WSTS, \emph{under well-typed substitutions} \\
\cline{2-3}
\end{tabularx}
\caption{Two independent pathways: the number of sessions/role-instances and the
freshness budget. Classical decidability occupies the ``sessions bounded''
column. The cell this work opens is the lower-right one, ``sessions unbounded,
freshness bounded''. Note that bounded freshness alone does not tame the
unbounded-message-size source of undecidability, so our result additionally
relies on the well-typed substitution discipline (Section~\ref{prelim}).}
\label{fig:landscape}
\end{figure}

%% file: bounded-freshness-meaning.tex

\paragraph{\textbf{What bounding freshness means}}
Because the phrase invites confusion with the familiar bounded/unbounded-session
dichotomy, we state what the assumption is. It is a bound  
\emph{separate}  from the number of sessions (see Figure~\ref{fig:landscape}),
and it can be read in three registers, as per the below:.

\emph{1. In terms of run generation.} Across all sessions and all agents of a role
$X$, at most $k_X$ distinct nonce-atoms are ever minted. The
$(k_X{+}1)$-th and every later session of role $X$ must reuse one of the $k_X$
values rather than mint a new one. The number of sessions stays unbounded; it is
the supply of fresh atoms that is capped.

\emph{2. In terms of attacker capability.} The Dolev--Yao attacker is not weakened
in its composition or derivation power: it still intercepts, replays, pairs and
encrypts without limit, over unboundedly many sessions. What saturates is the
stock of \emph{new atoms} it can chase. Beyond the terminal point of a run no
fresh atom appears, so the attacker's knowledge can only grow by recombining
terms it already has. Bounded freshness is thus a statement about atoms, not a
restriction on the attacker's moves; this is what distinguishes it from simply
handing the attacker a weaker rule set.

\emph{3. In terms of protocol design.} It models the realistic situation in which
nonce or key material is drawn from a bounded pool or regenerated on a cycle:
counters taken modulo $k$, rotated keys, finite entropy pools, cached
challenges. The security question then becomes the dual of the usual one. Rather
than ``bounded sessions over unbounded fresh material'', we ask whether secrecy
survives \emph{unbounded deployment} built on \emph{bounded fresh material}.

\noindent
The two examples in Section~\ref{sec:examples} make these three perspectives  concrete:
Example~\ref{sec:ex-freshness} shows the importance of freshness bounds when the
number of sessions left unbounded, and Example~\ref{sec:ex-param} evolves around what parameterisation per se brings or means (in).

%% file: table.tex
\begin{table}[H]
\centering
\small
\begin{tabularx}{\linewidth}{|l|X|X|}
\hline
\textbf{Aspect} & \textbf{Classical DY secrecy } & \textbf{Parameterised Secrecy } \\
\hline

\textbf{Main question} &
Secrecy: Does a given protocol have a leaky run? &
Parameterised Secrecy: Does a protocol-model of \emph{size} $n$ have a leaky run, for  all sizes $n$? (where $size$ is the number of role-instances) \\
\hline

\textbf{Baseline result} &
\textbf{Undecidable in general}, if either the number of nonces or the message-length are unbounded /
\textbf{Decidable}, if various bounds apply, e.g., $(T,k)$-bounded protocols~\cite{Tip05} &
\textbf{Undecidable in general}, due to DY derivations as well as the DY-attacker making parameterised models non-uniform~\cite{PVSurvey2016,AAMAS2016} / \textbf{Decidable} under bounded freshness + well-typed substitutions' assumptions (THIS WORK) \\
\hline

\textbf{View of sessions} &
Unbounded &
Unbounded and \textbf{parameterised by an arbitrary number of agents/role-instances} \\
\hline

\textbf{Attacker model} &
Full Dolev--Yao  &
\textbf{DY attacker restricted} to \textbf{well-typed substitutions}, i.e., excludes type-flaw attacks (THIS WORK) \\
\hline

\textbf{Role of bounded freshness} &
Alone does not yield a structural account of executions or scaling in number of sessions &
Induces \textbf{finite representative behaviour} via agent/term collapse \\
\hline

\textbf{Effect on executions} &
No general reduction from large to small systems &
\textbf{Large systems simulate smaller ones} via collapsing maps preserving DY derivability \\
\hline

\textbf{Proof methodology} &
Proof-theoretic or reduction-based  &
Using \textbf{well-structured transition systems (WSTS)}\cite{FinkelSchnoebelen01}, their coverability \\
\hline

\textbf{Structural insight} &
No general one on  how attacks scale with number of sessions &
\textbf{Cutoff theorem}: attacks in unbounded systems appear in bounded-size ones \\
\hline

\textbf{Algorithmic outcome} &
Either undecidability or bespoke bounded exploration &
\textbf{Finite exploration procedures} justified by cutoff or WSTS \\
\hline

\textbf{Interpretation of result} &  ``Classical-DY secrecy is decidable under  well-typed substitutions''.~\cite{RamSu03b} &
``Parameterised secrecy is decidable under bounded freshness per role and well-typed substitutions, i.e., then, secrecy for all session-counts reduces to finitely many representative systems'' (THIS WORK) \\
\hline

\end{tabularx}

\vspace*{0.1cm}
\caption{Deciding Classical Dolev--Yao Secrecy vs.\ This Work (Deciding Parameterised Dolev--Yao Secrecy)}
 \label{tab:intro}
\end{table}

%% file: relw.tex
A significant amount of related work was discussed in the introduction. We add and detail it now.

\subsection{Deciding and/or Verifying Dolev-Yao Security}

Deciding secrecy under the Dolev--Yao intruder has a long history. Early
undecidability results show that secrecy is undecidable when protocols allow
unbounded sessions, unbounded freshness, and unrestricted message growth~\cite{Dur99}.
Ramanujam and Suresh were among the first to recover decidability by bounding the
number of fresh nonces, under syntactic constraints such as tagging and
well-typedness~\cite{RamSu03}. Their results were extended in several directions.
D'Osualdo \emph{et al.}~\cite{DOsualdoOngTiu2017} decide secrecy for
depth-bounded processes via a structural-boundedness notion and WSTS; this is
conceptually close to ours but addresses a different axis of unboundedness and
yields no parameterised cut-off. Cortier, Delaune, and Sundararajan~\cite{CDS21}
identify a decidable class for reachability and equivalence based on type
compliance, without restricting the attacker but constraining the protocol class;
we do the complementary thing, constraining the attacker via well-typed
substitutions while leaving the protocol class unrestricted, and obtain a
\emph{uniform} result across all sizes. Froeschle~\cite{Froschle15} decides
leakiness for well-founded protocols by restricting execution shape; we instead
obtain termination through bounded freshness and parameterised collapse.
Equivalence-based results (Chr\'etien, Cortier, Delaune~\cite{CCD15}; dynamic
tagging by Arapinis, Delaune, Kremer~\cite{ADK14}) further highlight the role of
typing in taming complexity. Classically, Comon-Lundh and Cortier~\cite{Cortier}
established a \emph{small-model property} for secrecy and authentication, with
related complexity results in~\cite{MSR04,TipleaBEB08}.

\subsection{Parameterised verification}

 Parameterised
verification~\cite{PVSurvey2016} is well studied for distributed and concurrent
systems but has not been applied to symbolic protocols with a \emph{full}
Dolev--Yao intruder, one reason being that the DY intruder breaks the uniformity
that parameterised decidability arguments exploit~\cite{PVSurvey2016}. We bridge
this gap: under bounded freshness per role, parameterised DY secrecy is decidable
and admits a cut-off, combining ideas from the bounded-nonce
literature~\cite{RamSu03} with well-structured transition
systems~\cite{FinkelSchnoebelen01} through a uniform agent-collapsing preorder
that preserves Dolev--Yao derivability.

%% file: struct.tex
\bigskip

The rest of the  paper is structured as follows. Section~\ref{prelim} gives preliminaries. Section~\ref{sec:protocol-model} introduces a security-protocols' model that is parameterised on the number of roles, and defines the problem of  parameterised secrecy, in general.  Section~\ref{sec:examples} gives two worked examples: one that isolates our freshness setting, and one focused on the meaning of parameterisation. Section~\ref{sec:bounds}
introduces the concept of maps/collapsing-functions on agents which induce maps on terms as well as an order on the states of the protocol model.
 Section~\ref{sec:wqo} specialises the general protocol model, the maps and the order on states to the case of bounded freshness in the case of arbitrary many sessions, i.e., bounded number  of protocol nonces for unboundedly many agents.  Section~\ref{sec:decid} goes on to show that this latter order is key in deciding  parameterised secrecy in this case of bounded freshness.  For this problem, Section~\ref{sec:decision} gives two decision procedures, with asymptotic bounds on how many sessions one needs to consider to decide secrecy in our parameterised semantics and including complexity estimates. Section~\ref{sec:scope} delimits the scope and limitations of the approach.

%% file: dy-stuff.tex
We work in the symbolic (Dolev--Yao) model of cryptographic protocols
\cite{DY83}, where cryptographic primitives are assumed perfect and
attacker capabilities are captured by term construction and deduction.

\vspace*{0.0cm}
\paragraph{\textbf{Signature and terms}}
Let $\Sigma$ be a finite algebraic signature consisting of:
(i) atomic sorts for \emph{names}, \emph{nonces}, and \emph{keys};
(ii) constructors for pairing $\langle\cdot,\cdot\rangle$;
and (iii) cryptographic constructors, including symmetric encryption
$\mathsf{enc}_s(\cdot,\cdot)$ and public-key encryption
$\mathsf{enc}_{pk}(\cdot,\cdot)$.
Let $\mathcal{A}$ denote the  set of atomic terms and
$\mathcal{T}$ the set of ground terms freely generated from $\mathcal{A}$
using $\Sigma$.

\vspace*{0.0cm}
\paragraph{\textbf{Typing and substitutions}}
We assume a typing function $\tau : \mathcal{T} \to \Delta$ that assigns a
type to each term and is homomorphic over constructors.
A substitution $\sigma$ is a function assigning terms to terms and it models the attacker changing correct protocol messages
for arbitrary ones. A substitution is \emph{well-typed} if
$\tau(\sigma(x)) = \tau(x)$ for all variables $x$.
Only well-typed substitutions are allowed in our protocol executions, i.e., in our
attacker actions; this excludes ill-typed message forgery.

\vspace*{0.0cm}
\paragraph{\textbf{Dolev--Yao deduction}}
Attacker knowledge is represented by a set of terms $K \subseteq \mathcal{T}$.
Its deductive closure is defined using the standard Dolev--Yao inference rules,
which decompose messages (analysis) and build new ones (synthesis).
As per the expected, we use Paulson's two monotone operators~\cite{Paulson:1998}:
\begin{itemize}
  \item $\mathsf{analz}(K)$, closing $K$ under projections and decryptions
        when the appropriate keys are known;
  \item $\mathsf{synth}(K)$, closing $K$ under pairing and encryption.
\end{itemize}
The \emph{Dolev--Yao closure} of $K$ is the least fixed point ($lfp$)
\[
  \mathsf{DY}(K) \;\stackrel{\mathrm{def}}{=}\;
  \mathsf{lfp}\bigl(X \mapsto \mathsf{synth}(\mathsf{analz}(X))\bigr).
\]
A term $t$ is said to be \emph{derivable} by the adversary if
$t \in \mathsf{DY}(K)$.
Such a Dolev-Yao inference system, and related notions such as  normal forms of
$\mathsf{DY}$-derivations are given in full in Section~\ref{app:dy}.

%% file: dy-app.tex
We give the full inference rules for $\mathsf{analz}$ and $\mathsf{synth}$,
prove closure and normal-form properties, and we recall standard bounds on
``DY-derivation height''. We follow the notations in~\cite{SureshThesis}.

\subsubsection{DY synth/analz rules}

\paragraph{Signature, atoms, and terms}
Fix a simple algebraic signature $\Sigma$ with constructors for pairing $\pair{\cdot}{\cdot}$,
symmetric encryption $\encs{\cdot}{\cdot}$, public-key encryption $\encpk{\cdot}{\cdot}$,
hashing $\hash{\cdot}$, and atomic types for nonces, keys, names, etc.
Let $\Atoms$ denote the set of atoms, and let $\Term$ be the set of terms built
from $\Atoms$ using $\Sigma$.
Typing is as in the main text (well-typed substitutions only).

\paragraph{Intruder knowledge and closure.}
For any $X\subseteq \Term$, define two monotone operators
$\synth, \analz : 2^{\Term} \to 2^{\Term}$:
\[
\begin{array}{lcl l}
\synth(X) &=& X \cup
 \{\pair{u}{v} \mid u,v\in X\}
 \cup \\
 &  ~ & \{\encs{u}{k} \mid u,k\in X\} \cup \\
  &  ~ &  \{\encpk{u}{K} \mid u,K\in X\} \cup \\

   &  ~ & \{\hash{u} \mid u\in X\},\\[0.3em]
\analz(X) &=& X \cup
 \{u \mid \exists v.\ \pair{u}{v}\in X\}
 \cup \\
  &  ~ &  \{v \mid \exists u.\ \pair{u}{v}\in X\}\\
&&\quad \cup\{u \mid \exists k.\ \encs{u}{k}\in X \wedge k\in X\} \\
&&  \cup\{u \mid \exists K^{-1}.\ \encpk{u}{\pk(K)}\in X \wedge \\
&& ~~~~~~~~~~~~~~~~~K^{-1}\in X\}.
\end{array}
\]

Hashes are one-way: $\hash{u}$ never yields $u$ in $\analz$.
The Dolev-Yao closure of $X$ is
\[
\know(X) := \mathrm{lfp}\bigl(Y \mapsto \synth(\analz(Y))\bigr)
\quad\text{with }X\subseteq Y.
\]

\paragraph{Inference rules.}
Equivalently, $\know(X)$ is the set of terms derivable by the following rules
(starting from hypotheses $X$).

\scalebox{0.9}{
\begin{mathpar}

\inferrule*[right=Ax]{~}{\Gamma,\,u \vdashDY u}

\inferrule*[right=Pair]{\Gamma \vdashDY u \\ \Gamma \vdashDY v}
{\Gamma \vdashDY \pair{u}{v}}
\and
\inferrule*[right=Fst]{\Gamma \vdashDY \pair{u}{v}}
{\Gamma \vdashDY u}
\and
\inferrule*[right=Snd]{\Gamma \vdashDY \pair{u}{v}}
{\Gamma \vdashDY v}

\inferrule*[right=SEnc]{\Gamma \vdashDY u \\ \Gamma \vdashDY k}
{\Gamma \vdashDY \encs{u}{k}}
\and
\inferrule*[right=SDec]{\Gamma \vdashDY \encs{u}{k} \\ \Gamma \vdashDY k}
{\Gamma \vdashDY u}

\inferrule*[right=PKEnc]{\Gamma \vdashDY u \\ \Gamma \vdashDY \pk(K)}
{\Gamma \vdashDY \encpk{u}{\pk(K)}}
\\\vspace*{0.6em}
\and
\inferrule*[right=PKDec]{\Gamma \vdashDY \encpk{u}{\pk(K)} \\ \Gamma \vdashDY K^{-1}}
{\Gamma \vdashDY u}

\inferrule*[right=Hash]{\Gamma \vdashDY u}{\Gamma \vdashDY \hash{u}}

\end{mathpar}
}

\subsubsection{Normal-form DY Proofs } \label{app:ex}

The material in this subsection is used in Lemma~\ref{lem:normalisation}
and, more importantly, in Proposition~\ref{prop-bound},
where bounded-height normal DY derivations provide a computable
bound on the length of leak witnesses after terminality.

We now describe aspects linked to normal-form DY proofs, in the style of~\cite{SureshThesis}.

\begin{lemma}[Normal-form theorem~\cite{SureshThesis}] \label{cfp}
Fix the signature $\Sigma$ and a finite set of atoms $A$.
There exists a computable function $g$ depending only on $\Sigma$
such that for any $K\subseteq \Term$ and any ground term $t$:

$t\in \know(A\cup K)
\quad\Longleftrightarrow $ \\
$\exists\text{ cut-free proof of }A\cup K\vdashDY t
\text{ of height }\le g(|A|)$
\end{lemma}
See ~\cite{SureshThesis} for proof.

\medskip

To understand this, consider the following simplified situation, as it could occur at or after a
terminal state in our parameterised system.
Let the intruder initially know
\[
K_0 = \{\,a,\; b,\; k,\; \encs{\langle a,b\rangle}{k}\,\}.
\]
We ask whether the atomic message $a$ can be derived, i.e.\
whether $K_0 \vdash a$.

\smallskip
\noindent
A minimal DY derivation $\mathcal{D}_0$ proceeds as follows:
\begin{enumerate}
\item From $\encs{\langle a,b\rangle}{k}$ and $k$, apply rule \textsf{SDec} to
obtain $\langle a,b\rangle$.
\item From $\langle a,b\rangle$, apply \textsf{Fst} to obtain $a$.
\end{enumerate}
The corresponding proof tree is:

\begin{mathpar}
\inferrule*[right=SDec]{
  \inferrule*[right=Ax]{ }{K_0 \vdash \encs{\langle a,b\rangle}{k}} \\
  \inferrule*[right=Ax]{ }{K_0 \vdash k}
}{
  K_0 \vdash \langle a,b\rangle
}
\\
\inferrule*[right=Fst]{
  K_0 \vdash \langle a,b\rangle
}{
  K_0 \vdash a
}
\end{mathpar}

\noindent
This proof has height~$2$ and is already cut-free: both inference chains are
strictly decreasing in the measure ``encryption depth plus pairing depth.''

\smallskip
\noindent
In non-normal form, one might interleave synthesis and analysis, e.g.\
constructing $\langle a,b\rangle$ from $a$ and~$b$ and then decrypting again, giving
a redundant loop of height~$4$:
\[
\textsf{Pair} \;\circ\; \textsf{SDec} \;\circ\; \textsf{Pair} \;\circ\; \textsf{Fst}.
\]
The normalisation procedure collapses these unnecessary
synth/analz alternations to the minimal form above, whose height is bounded by
$g(|A|)=2$ for this atom basis $A=\{a,b,k\}$.

\smallskip
\noindent
In Proposition~\ref{prop-bound}, this height bound translates to a distance bound along any
run $\pi$:
if $s_\tau$ is the terminal state where $\encs{\langle a,b\rangle}{k}\in s_{\tau,I}$,
then within at most $c\!\cdot\!g(|A|)=2c$ transitions the intruder learns~$a$,
i.e.\ the leak occurs by step $\tau+2c$.

\smallskip
\noindent
This illustrates concretely how a \emph{cut-free DY proof of bounded height}
yields a uniform finite witness bound $k=f(|Pr|,\#\text{bound})$
for all leaky runs in the Dolev-Yao semantics.

%% file: wsts-b.tex
We briefly recall \emph{well-structured transition systems} (WSTS), a generic
framework for decidability in infinite-state systems that underlies our analysis;
we follow~\cite{FinkelSchnoebelen01}.

A \emph{quasi-order} (qo) is a reflexive, transitive relation $\preceq$ on a set
$S$; it is a \emph{well-quasi-order} (wqo) if every infinite sequence
$s_0,s_1,\ldots$ contains indices $i<j$ with $s_i \preceq s_j$ (equivalently, no
infinite descending chains and no infinite antichains). A set $U \subseteq S$ is
\emph{upward-closed} if $s \in U$ and $s \preceq s'$ imply $s' \in U$.

A transition system $(S,\rightarrow)$ equipped with a wqo $\preceq$ is a
\emph{WSTS} if $\preceq$ is \emph{upward compatible} with $\rightarrow$: whenever
$s \preceq s'$ and $s \rightarrow t$, there is $t'$ with $s' \rightarrow^{*} t'$
and $t \preceq t'$. Intuitively, larger states $s'$ simulate smaller ones $s$. A
WSTS is \emph{effective} if $\preceq$ is decidable and predecessors of
upward-closed sets have a computable finite basis. The \emph{coverability problem}
asks whether some $t' \succeq t$ is reachable from $s$; in effective WSTS it is
decidable by backward reachability from $\uparrow t$, which stabilises by
well-quasi-ordering. This mechanism is central to our decidability and cut-off
arguments~\cite{FinkelSchnoebelen01,PVSurvey2016}.

%% file: prot-model.tex
\subsection{Introduction to our Protocol Model}
Our basis is  a commonplace, symbolic protocol model and the trace-based semantics of \cite{RamSu03}, as summarised in this subsection.

\paragraph{\textbf{Protocol Descriptions \& Roles}}
A \emph{protocol} consists of a finite set of \emph{roles}, denoted
$A,B,\ldots$.
Each role is given by a finite sequence of \emph{role actions} of the form $$ A \rightarrow B : t,$$
where $A$ is the sender \emph{role}, $B$ the intended receiver role, and
$t \in \Term$ is a symbolic message modelled as a term with atomic subterms or atoms.

Certain atoms that ``originate'' from or pertain to the role (e.g., nonces generated by that role and used as key material, long-term keys) are \emph{secret to the role} and denoted \emph{$Secret_{X}$} for a role $X$. All such secrets, i.e., $\bigcup_{X \in \Roles} Secret_X$ is denoted \emph{$Secret$}.

\paragraph{\textbf{Agents and Role-Instantiations}}
In protocol executions, an {unbounded} number of \emph{agents}  instantiate each role.
We write \emph{$i_A$}  for the $i$-th
agent executing role $A$, or simply \emph{agent $i$} when the role is clear.
Instantiating a \emph{role action $A \rightarrow B : t$} by agents
$i_A$ and $j_B$ yields \emph{concrete actions}
\[
  i_A! j_B : t
  \qquad\text{and}\qquad
  j_B?\,\cdot : t,
\]
representing respectively the emission and reception of the message. The \emph{``$\cdot$''} anticipates the usual Dolev-Yao semantics whereby all message received are mediated by the attacker (see  paragraph ``Transitions'' in Subsection~\ref{actions}, for details).

\underline{Note}:  We work under the commonplace setting in security-protocol verification  that agents are \emph{uniform} instantiations of roles (i.e., identical up to renaming, per role), and that they are not aware of their own indexing (i.e., agent $i$-th of role $X$ does not know it is the $i$-th) and their indexing is not used as part of the functionality of the protocol.  In parameterised-verification terminology~\cite{FinkelSchnoebelen01}, role instantiations are homogeneous.

\paragraph{\textbf{Notation for actions}}
\label{par:action-notation}
For a concrete action $a$ we write $a.\mathrm{ag}$ for the acting agent and
$a.\mathrm{step}$ for the position of $a$ in that agent's role script; and we
write $\epsilon$ for the empty action sequence. These are used when we lift maps
to runs in Section~\ref{sec:bounds}.

\paragraph{\textbf{Substitutions and typing}}
Variables occurring in role messages are instantiated at execution time
by substitutions,  modelling a Dolev-Yao attacker's intervention onto all communication channels as we detail in  Section~\ref{sec:protocol-semantics} below.
As we said in Section~\ref{prelim}, only \emph{well-typed substitutions} are permitted.

\paragraph{\textbf{Intruder model}}
We assume a single \emph{distinguished agent $I$} representing  the  Dolev-Yao intruder.  \emph{Intruder knowledge} is represented as a set \emph{$K_{I} \subseteq \Term$} and
evolves monotonically according to the Dolev--Yao deduction system detailed in Section~\ref{app:dy}.
We write \emph{$K \vdashDY t$} when $t$ is derivable from $K$, i.e.,
$t \in \know(K)$.
The normal, Dolev-Yao semantics is detailed below in Section~\ref{actionssec}.

\subsection{Protocol Semantics and Parameterised Security Systems}
\label{sec:protocol-semantics}

We now give our operational  semantics of protocols, i.e., we  define a parameterised transition-system model encoding protocol executions.

\paragraph*{\textbf{Local and Global States}}

The \emph{local state of an honest agent $i$} is a pair  \emph{$s_i = (\mathsf{step}_i, K_i)$},
where $\mathsf{step}_i$ denotes the current position in the role
specification, and $K_i \subseteq \Term$ is the set of terms known by
agent $i$.
The \emph{local state of the intruder $I$} is
  \emph{$ s_I = K_I$},
where $K_I \subseteq \Term$ is the intruder's knowledge.

For a fixed system size
  $ n = (n_A)_{A \in \mathsf{Roles}}$,
a \emph{global state} is a  tuple \emph{\[s = ((s_{1_A},\ldots,s_{n_A} \mid A \in \mathsf{Roles}),\; s_I),\]} which, for all roles $A$ (i.e., $A, B , C, \ldots$), for all honest agents $i_A$ of that role (i.e., $1_A, 2_A, \ldots, 5_B, \ldots, 7_C, \ldots$),  collects their local states,   and that of the intruder's.

\subsection{Runs and Systems of Size \texorpdfstring{$n$}{n}} \label{actionssec}

\emph{Admissible initial states}  encode a setup for the protocol start (i.e., agents' steps set to $0$), and include arbitrary well-typed instantiations of role
variables that originated from that role, particularly: well-typed instantiations for $Secret_{X}$ for (any number of) agents of role $X$, next to any  constants,  and well-typed instantiations of  $Secret_{I}$.

\paragraph{\textbf{Transitions}} \label{actions}
A \emph{transition} is an action $\alpha$ taking place at system-state $s$: \emph{$s \xrightarrow{\alpha} s'$}, where $\alpha$  is a  send action $ i_A! j_X : t,
  ~\text{or } \text{ a receive action }\\
  j_B?\,\cdot : t$.

The enablement of transitions and state updates are in line with~\cite{RamSu03}.  %
\\ -- A \emph{send action $i_A! \cdot : t$}
is \emph{enabled} at a state $s$ if: \\
~~~ (a) agent $i_A$ is in the corresponding protocol step for sending $t$ and can compose term $t$  at its state $s_{i_{X}}$, or  \\
~~~ (b) if message $t$  derivable from the current intruder's knowledge $s_{I}$
via a well-typed substitution.

In the first case $(a)$, $t$ is  added   to the intruder's knowledge and  the $step$ variable of agent $i$ is updated. \\
In the second case $(b)$, $t$ is added to the knowledge of a recipient agent $j$.  %
\\-- A \emph{receive action $\cdot! j_X : t$}
updates the state of agent $j_X$ if $t$ is match-received~\cite{RuTu03} to the state of $s_{j_{X}}$, i.e., the subterms of $t$ that are already in $s_{j_{X}}$ match those in  $t$. \\
~~~ The update consists of adding decomposed subterms of $t$ that are not already in $s_{j_{X}}$ to the new state $s'_{j_{X}}$,  and the step variable of agent $j$ is updated.

\paragraph{\textbf{Runs}} A \emph{run} is a  sequence of transitions
  $$s_0 \xrightarrow{\alpha_1} s_1 \xrightarrow{\alpha_2} \cdots \xrightarrow{\alpha_k} s_k \ldots,$$
starting from an admissible initial state $s_0$.
Runs are traces in the sense of \cite{RamSu03}, with intruder
interception implicit at every send, as the semantics of enabled transitions $s \xrightarrow{\alpha} s'$  above explained.

\paragraph{\textbf{Systems of Fixed Size}}
For a fixed \emph{$n=(n_A, n_B, \ldots)_{A, B, \ldots \Roles}$}, the \emph{system of size $n$}, denoted \emph{$\mathbb{S}(n)$}, is the set of all global reachable states together with all possible transitions
 from {admissible} initial states instantiating $n$ such agents.

\subsection{Leaky Runs and Parameterised Secrecy} \label{sec:leakage}

\paragraph{Leakiness} \label{lr}
A run is \emph{leaky} if there exists a state $s$ along the run and a
secret atomic term $a$  which is learnt by the Intruder, i.e.:
\[
  s_0 \xrightarrow{\alpha_1} s_1 \xrightarrow{\alpha_2} \cdots
  \xrightarrow{\alpha} s \cdots \text{ and } a \in Secret \cap  K_{I}^s,
\]
where \emph{$K_{I}^s= {DY}(s)$}.

\begin{definition}[Parameterised secrecy] \label{defpsp}
The \emph{\textbf{parameterised secrecy problem (PSP)}} asks whether {for an arbitrarily fixed } $n,$ there exists a leaky run in the parameterised system $\mathbb{S}(n)$ of size $n$.

~

Equivalently, a protocol is \emph{parameterised-secure} if no leaky run exists
for any system size.
\end{definition}

%% file: example-freshness.tex

\subsection{Example 1: One Protocol, Two Freshness Regimes}
\label{sec:ex-freshness}

We begin with a deliberately small, made-up protocol whose sole purpose is to
isolate our \emph{freshness} settings. The number of sessions is left unbounded
throughout; only the freshness assumption changes. The example shows three
things at once: that bounded freshness genuinely \emph{changes} the secrecy
question, that it nonetheless keeps it \emph{hard} (the witness is not of
constant size), and that our collapse does not manufacture spurious attacks.

\paragraph{\textbf{The protocol $\Pex$}}
There is one initiator role $A$ and one responder identity $R$, which $A$ does
\emph{not} authenticate; in a Dolev--Yao execution $R$ may therefore be an
honest responder $B$ or the intruder $I$. Each agent $i_A$ of role $A$ owns a
single nonce $n_{i}$ (of type \textsf{nonce}), subject to the per-role freshness
bound $k_A$. Each honest responder $B$ owns a secret $s_B \in \Secret$. The
protocol has two role-actions,
\[
\begin{array}{ll}
  (1)~ A \rightarrow R : & \encpk{n_{A}}{\pk{R}},\\[2pt]
  (2)~ R \rightarrow A : & \encs{s_{R}}{n_{A}},
\end{array}
\]
where $\pk{R}$ is $R$'s public key (public), and in action $(2)$ the honest
responder returns its own secret keyed \emph{symmetrically} under the nonce it
just received. The datum to be kept secret is $s_B$ (equivalently, any honest
responder's secret). Only the constructors $\mathsf{enc}_{pk}$,
$\mathsf{enc}_{s}$ and atomic keys are used, so $\Pex$ is within our considered class of  protocols.

\paragraph{\textbf{Regime 1: unbounded freshness}}
Suppose each $A$-agent draws a truly fresh nonce, so distinct agents own
distinct nonces. We claim $\Pex$ is parameterised-secure: $s_B$ is not derivable
for any system size $n$.

\begin{quote}
\emph{Invariant.} In any reachable state, the only network terms are of the form
$\encpk{n_i}{\pk{R_i}}$ and $\encs{s_{R_i}}{n_i}$. A nonce $n_i$ enters
$\know(K_I)$ only if some session sealed it to $R_i = I$, that is, only if
agent $i$ addressed the intruder. For an honest-responder session $i$ (with
$R_i = B$), the nonce $n_i$ is encrypted under $\pk{B}$ and appears nowhere else,
so $n_i \notin \know(K_I)$; hence $\encs{s_B}{n_i}$ cannot be opened and
$s_B \notin \know(K_I)$.
\end{quote}

\noindent
Because under unbounded freshness the nonce of an honest-responder session is
distinct from every intruder-facing session's nonce, knowledge gathered from
intruder-facing sessions is useless against honest ones. Secrecy holds
uniformly in $n$.

\paragraph{\textbf{Regime 2: bounded freshness}}
Now impose the per-role bound $k_A$. With more than $k_A$ agents of role $A$, the
canonical collapse forces two agents to share a nonce (Definition~\ref{cdf}).
This is exactly the case the invariant above excluded. Consider two $A$-agents
whose nonces coincide at the atom $n$: agent $i$ addressed the intruder
(so $n \in \know(K_I)$ after action $(1)$ of session $i$), while agent
$i'$ addressed an honest $B$ (so the network carries $\encs{s_B}{n}$ after
action $(2)$ of session $i'$). The intruder now decrypts and learns $s_B$: the
run is leaky.

The witness is not of constant size. To \emph{force} a collision between an
intruder-facing nonce and an honest-facing one, one needs, by pigeonhole,
$k_A + 1$ agents of role $A$ (together with a single honest $B$). Thus the
minimal leaky system has size $\Theta(k_A)$ in role $A$: the required number of
sessions equals the freshness bound, not some fixed constant. This is precisely
why bounded freshness is \emph{not} the same as bounding the number of sessions,
and why the cut-off in Section~\ref{sec:decid} lands on $k_A$ rather than on an
arbitrarily chosen small number.


%% file: example-param.tex

\subsection{Example 2: Parameterisation as a Natural Security Question}
\label{sec:ex-param}

If Example~\ref{sec:ex-freshness} isolated freshness, we now look at 
\emph{parameterisation}. Freshness is held fixed (bounded per role), and the
quantity that varies is the size of the population, which here is genuinely part
of the security question rather than a modelling convenience.

\paragraph{\textbf{The setting}}
Consider a server-based key-distribution protocol $\Pksd$ of Needham--Schroeder
shared-key style. There is a trusted server $S$ that shares a long-term
symmetric key $K_{XS}$ with every principal $X$ drawn from a population of size
$n$. An initiator $A$ obtains a fresh session key $k_{AB}$ for a responder $B$:
\[
\begin{small}
\begin{array}{ll}
  (1)~ A \rightarrow S : & \pair{A}{\pair{B}{n_A}},\\[2pt]
  (2)~ S \rightarrow A : & \encs{\pair{n_A}{\pair{B}{\pair{k_{AB}}{\encs{\pair{k_{AB}}{A}}{K_{BS}}}}}}{K_{AS}},\\[2pt]
  (3)~ A \rightarrow B : & \encs{\pair{k_{AB}}{A}}{K_{BS}}.
\end{array}
\end{small}
\]
Only symmetric encryption, pairing, nonces and atomic keys occur, so $\Pksd$ is
within our fragment. The property of interest is secrecy of the session key
$k_{AB}$ between two honest principals.

\paragraph{\textbf{Why the parameter is intrinsic}}
The question a designer actually cares about is not ``is $k_{AB}$ secret in a
system of two principals?'' but ``does $k_{AB}$ remain secret however large the
user base grows, and whichever other principals, honest or compromised, are
present?''. In our terms, each additional principal is one more agent
instantiating role $A$ or role $B$, so the population size $n$ is exactly the
parameter of the family $\Sys(n)$. Parameterised secrecy asks whether a leaky
run exists for \emph{some} $n$; parameterised security asserts secrecy for
\emph{every} $n$. The number of role instances is here a first-class security
variable, not an artefact.

\paragraph{\textbf{What the cut-off brings}}
Verification practice analyses a handful of principals with a tool and then
trusts that the conclusion carries over to an unbounded deployment. Under
bounded freshness and well-typed substitutions, our cut-off theorem
(Theorem~\ref{thm:cutoff}) turns that folklore into a proven,
concrete bound: secrecy of $k_{AB}$ across \emph{all} population sizes follows
from its secrecy in systems with at most $c_A, c_B, c_S$ agents per role, where
$c_X = k_X$. A larger population contributes only further agents of the same
roles, and any such agent collapses onto one of the canonical representatives
under $h_c$ while preserving the intruder's derivations. Hence no population
size beyond the cut-off can exhibit a secrecy violation that a bounded one does
not already witness. Here freshness never varies, only the population does: the
example shows that the parameterised reading is the natural one for real
key-distribution deployments, and that the cut-off supplies the missing
justification for the standard small-instance verification methodology. Taken
with Example~\ref{sec:ex-freshness}, the two examples separate the axes cleanly.

%% file: boundedness.tex
In this section, we introduce the idea of maps/collapsing-functions on agents which lift to maps on terms and actions, and, with that, induce an order on the states of the protocol model.
The order over states has to be such that it preserves DY-derivability up to the terms' collapsing.  Figure~\ref{fig:bos-canonical-simple} summarises these generic concepts of reducing a larger system to a smaller one, in a way that the order over state preserves DY-derivability.

\subsection{Good Agent Maps}
\label{subsec:good-agent-maps}

Let $\Ag$ be the (countably infinite) set of agents, partitioned by roles.
Recall that each agent $i_X \in \Ag$ instantiates a unique role $X$.

\begin{definition}[Good agent map] \label{def1map}
A function
 \emph{$ h : \Ag \rightarrow \Ag
$}
is a \emph{good agent map} if it satisfies the following conditions:
\begin{enumerate}
  \item \emph{Role-respecting:} for all agents $i_X \in \Ag$,
        $h(i_X)$ is an agent instantiating the same role $X$;
  \item \emph{Surjective but not injective:} every agent has at least one
        preimage under $h$, and at least two distinct agents are mapped to
        the same image.
\end{enumerate}
\end{definition}

Intuitively, a good agent map collapses multiple agents of the same role
into a smaller set of representative agents, while preserving role
structure.
Such maps will be used to relate executions of larger systems to
executions of smaller ones. An image showing this  is given in Figure~\ref{fig:bos-canonical-simple}, on page~\pageref{cdf}.

\subsection{Induced Term Maps}
\label{subsec:induced-term-maps}

A good agent map induces a corresponding transformation on terms.

\begin{definition}[Induced term map] \label{def2map}
Let $h : \Ag \to \Ag$ be a good agent map.
The  \emph{induced term map},
  \emph{$\hat{h} : \Term \rightarrow \Term$}
is defined by uniformly replacing, in any term $t$, all agent-indexed
atomic subterms pertaining to an agent $i$ by the corresponding atomic
terms pertaining to $h(i)$.
The map $\hat{h}$ is extended homomorphically to sets $2^{\Term}$ of  terms.
\end{definition}

Note that the induced map $\hat{h}$ preserves typing and the structure of
cryptographic constructors.

\subsection{Lifting Good Agent-Maps to Maps on Actions} \label{sec:lift-runs}
We now show how, naturally, we lift maps from agents and terms to actions.

\begin{definition}[Maps on Actions] \label{def:action-lift} Let $h:\Ag\to\Ag$ be a good agent map.

For a concrete action $a$, either a send $i!\cdot:t$ or a receive $i?\cdot:t$,
define $\tilde{h}(a)$ by
\begin{align*}
  \tilde{h}( i!\cdot:t ) &:= h(i)!\cdot:\hat{h}(t),\\
  \tilde{h}( i?\cdot:t ) &:= h(i)?\cdot:\hat{h}(t).
\end{align*}

\smallskip

Now, we also make explicit the extension of good agent maps from actions to finite runs.
We extend $\tilde{h}$ to finite action sequences $\pi \in Act^*$ inductively as follows.
Let $\pi = aa'\pi'$ with $a,a'\in Act$.

\begin{align*}
  \tilde{h}(\epsilon) &:= \epsilon,\\
  \tilde{h}(aa'\pi') &:= \tilde{h}(a)\tilde{h}(\pi')
    && \text{if } h(a.\mathrm{ag}) = h(a'.\mathrm{ag})
       \text{ and } \\
   ~& ~    && a.\mathrm{step} = a'.\mathrm{step},\\
  \tilde{h}(aa'\pi') &:= \tilde{h}(a)\tilde{h}(a')\tilde{h}(\pi')
    && \text{otherwise}.
\end{align*}

\smallskip
\end{definition}

The definition above says that, for a concrete action
$\act$ performed by an honest agent $i$ carrying a  term $t$, the induced map
$\widetilde{h}(\act)$ replaces agent $i$ with $h(i)$ and the term $t$ with $\hatmap{h}(t)$.
When two consecutive actions map to the same agent/step pair, the induced map allows
\emph{stuttering collapse}, by deleting duplicates (so, in other words, $\widetilde{h}$ can map
finite paths to single steps).

\begin{definition}[Maps on runs]
\label{def:run-lift}
Let $h:\Ag\to\Ag$ be a good agent map.

\begin{itemize}
\item For a single labelled action $\act$ performed by an honest agent $i$ with
payload $t$, define $\widetilde{h}(\act)$ as in
Definition~\ref{def:action-lift}.

\item Extend $\widetilde{h}$ to finite action sequences (runs)
$\pi=\act_1\act_2\cdots\act_k\in \mathrm{Act}^\ast$ by
\[
\widetilde{h}(\pi)
\;:=\;
\mathrm{stut}\bigl(\widetilde{h}(\act_1)\cdot
\widetilde{h}(\act_2)\cdots\widetilde{h}(\act_k)\bigr),
\]
where $\mathrm{stut}(\cdot)$ removes consecutive duplicate
agent/step pairs (stuttering collapse).
\end{itemize}
\end{definition}

\subsection{Our Preorder on Protocol-Model States}
\label{subsec:bos-relation}

We now define the preorder on system's states induced by the maps above.

\begin{definition}[Bound-based ordering (BOS)] \label{bosdefinition}
Let $s$ and $s'$ be two global states. We say they are in a \emph{BOS order} and we write
\[
  s \;\preceq\; s'
\]
if there exists a good agent map $h$ such that:
\begin{enumerate}
  \item for every honest agent $i$ in $s$, there exists an agent $j$ in
        $s'$ with the same role and
        \[
          \hat{h}(s'_j) = s_i;
        \]
  \item the intruder knowledge is preserved:
        \[
          \hat{h}(K_I^{s'}) = K_I^{s}.
        \]
\end{enumerate}
\end{definition}

Intuitively, $s \preceq s'$ holds when $s'$ can be obtained from $s$ by duplicating agents
and renaming them in a role-consistent manner. The BOS relation captures the
idea that larger systems simulate smaller ones modulo agent collapse.

Importantly, for a \emph{general} good map, BOS is defined in terms of equality of collapsed knowledge
$\hat{h}(K_I^{s'})= K_I^s$, and does not by itself imply equivalence of Dolev–Yao derivability.
Rather, derivability is preserved under $\hat{h}$ in the forward direction,
as established in Section~\ref{dy-fw}. (For the \emph{canonical} map committed to in Section~\ref{sec:wqo}, which relabels agents over a shared canonical atom set and hence identifies no distinct atoms, preservation holds in both directions; see Section~\ref{par:canonical-atoms}.)

~

\underline{Note.}  The maps $h$  and $\hat{h}$ of agents and terms  are so far abstract and have been defined independently of the size of the systems or of any bounds or restrictions. Many concrete maps $h$ may exist.
However, for the BOS order $\preceq$ to be as per  Definition~\ref{bosdefinition} (i.e., the order preserve the DY knowledge modulo the term-map), it is non-trivial to find  right concrete maps $h$  and $\hat{h}$.

%% file: wqo.tex

With the note above in mind, we now give a  specific instance of the agents' map which induces a specific BOS order on states. Later, we will prove that this specific order is a well quasi-order (wqo), providing the structural basis for the decidability argument developed in Section~\ref{sec:decid}.

\subsection{Nonce Bounds and Bounded Initial States}
\label{sec:nonce-bounds}

To give our specific map $h$ of agents and induce our specific order $\preceq$ on states, we set a systematic and parameterised bound on the number of nonces used in protocol executions. So, this subsection introduces this  ``bounded freshness'' assumption and other mathematical objects linked to it.

We parameterise
freshness by explicit \textbf{per-role bounds} and use these bounds to constrain the
family of admissible initial states for fixed-size systems $\mathbb{S}(n)$.

\begin{definition}[Freshness bound per role]
Let $\Roles$ be the finite set of protocol roles.
For each role $X \in \Roles$, a fixed \emph{positive natural number $k_X$} is called the
\emph{freshness bound of role $X$}.
\end{definition}

Intuitively, across any execution in a system of fixed size, each agent
instantiating role $X$ may generate at most $k_X$ fresh nonce-atoms of type
\textsf{nonce}.
For now, assume each bound $k_X$ to be super-exponential in the protocol-role size (i.e., in the number of atomic terms in the  description of the role $X$).

 This mechanism ensures that, although the number of agents is unbounded, the
set of distinct atoms that can arise in any run is globally bounded per role.

\subsection{Bounded Initial-State Families for Fixed Size}

As in Section~\ref{sec:protocol-model}, let $\mathbb{S}(n)$ be  system of size  $ n = (n_X)_{X \in \Roles}$,
where any arbitrarily large $n_X$ is the number of agents instantiating role $X$.

To enforce bounded freshness, we do \emph{not} pick a single initial state.
Instead, for each size $n$, we consider a \emph{(possibly infinite) set of admissible
initial states}, written \emph{$\mathbb{S}(n).\mathsf{Init}$}, where each $s_0 \in \mathbb{S}(n).\mathsf{Init}$
assigns:
\begin{itemize}
  \item to each honest agent $i_X$, a local state $(\mathsf{step}=0,K0_{i_X})$,
        where $K0_{i_X}$ contains the role-originated  public/shared atoms instantiated by
        well-typed substitutions
        and leaves role variables unspecified (to be instantiated later  by
        well-typed substitutions);
  \item to the intruder a knowledge set $K0_I$ containing the standard public
        information and any globally public keys, and their own fresh nonces instantiated by
        well-typed substitutions.
\end{itemize}

\input{concreteFig}

Importantly, $\mathbb{S}(n).\mathsf{Init}$ is generated \emph{subject to the bounds}
$k_X$:  the total number of  fresh nonces that can originate
from all agents $i_X$ during any run starting from $s_0$ needs to add up to  most $k_X$
distinct atoms. If the number of agents of role $X$ exceeds $k_X$, then admissible initial states
are constructed so that these agents reuse previously allocated nonce atoms.
Formally, this enforces that the set of distinct nonce atoms originating from
role $X$ is bounded by $k_X$.

\vspace*{0.1cm}
\subsubsection*{\textbf{A Preamble to the Rest of the Section}}

\paragraph{\textbf{ Finite and Infinite Structures  in Bounded Freshness}}\label{terminalityp}
A system  $\mathbb{S}(n)$ has an infinite set of initial states, for each size $n$, but each state $s_0 \in S(n).\mathsf{Init}$   has a bounded level of freshness for honest agents.
 Also, there is no mixing of systems of sizes $p$ and $q$, where  $p \neq q$.
 This means that the number of  fresh sessions to be created within any system of  any size $n$ is bounded and ``set'' apriori.

Further, because each run starting from $s_0 \in \mathbb{S}(n).\mathsf{Init}$ involves only the
finite set of agents present in size $n$, and each such agent has a bounded
freshness budget, every run in $\mathbb{S}(n)$ has what we call a \emph{``terminal point''}. At this point,  no honest agent can
perform further freshness-producing actions. The intruder may continue
to derive composite terms from already-known atoms, but no new atoms can be
introduced by honest parties once all budgets are exhausted.  Thus, although the set of admissible initial states is infinite, every run
evolves over a finite set of atoms and can therefore be represented finitely; this is shown formally in Lemma~\ref{lem:terminality}.

\paragraph{\textbf{Canonical atoms: a key consequence}}
\label{par:canonical-atoms}
We state here one consequence of the construction that is central to reading
all the maps in this section correctly, and which we use repeatedly below.
Because admissible initial states reuse atoms once a role's budget $k_X$ is
exhausted, and because no run mints atoms beyond those budgets, \emph{every
reachable state of every system, of any size, uses only the finitely many
canonical atoms allocated per role}. Consequently, the induced term map
$\hatmap{h}$ of the canonical collapse (Definition~\ref{cdf} below) acts on
\emph{agent identities and their indexing}, not by inventing new atoms: two
agents in the same canonical class already carry the same canonical atoms.
This is why, for two states related by the canonical order, collapsing never
changes the underlying atom set, and the direction of the maps is unambiguous:
$\hatmap{h}$ sends a \emph{larger} state (more agents) to a \emph{smaller} one
(fewer representative agents) over the \emph{same} atom set.

\paragraph{\textbf{Why Freshness Bounds Matter}}

The bounds $(k_X)_{X\in \Roles}$ induce a canonical ``folding'' of agents of each
role into $k_X$ representative classes.
These maps in turn define a preorder over system states, enabling our final results later, i.e., essentially, using  WSTS coverability~\cite{FinkelSchnoebelen01}, to prove parameterised secrecy in the setting of bounded freshness.

\subsection{A Canonical Agent Map \& Its Induced Well-Quasi-Order (WQO) over States}
 \label{sec:bos-wqo}

We now give a specific agent-folding map $h$, under which we prove our results.

\begin{definition}[Canonical good agent map] \label{cdf}
We fix per-role freshness bounds $(k_X)_{X\in \Roles}$. Then,
a \emph{canonical good agent map $h_c$} is:
\[
h_c(i\_X) \;:=\; (\,1 + (i-1 \bmod k_X)\,)_{X},\qquad h_c(I)=I.
\]
\end{definition}

This map is illustrated also on Figure~\ref{fig:bos-canonical-simple}. Such a map simply says that if a system/protocol-execution is large in its size $n$ (i.e., in  number of agents) so that the number $k_X$ of nonces allocated to role $X$ has been ``used'', then some agent $j> k_X+1$ of that role $X$ will start reusing previously allocated nonce atoms for that role in a canonical,
index-based manner.
That is to say, the initial states of that agent will be set as such via $\mathbb{S}(n).Init$.

Its induced term map $\hatmap{h_c}$ is as per Section~\ref{sec:bounds}, and its $\preceq_c$ relation on states\footnote{We have to prove that this induced $\preceq_c$ is a BOS, i.e., it preserves DY-derivability in the forward direction. This is done later in Lemma~\ref{lem:forward-dy}.} as per Definition~\ref{bosdefinition}.

\medskip

\underline{Note:} From here on, in the technical results that follow, we refer just to these canonical maps $h_c$ and  $\hatmap{h_c}$ and to their induced BOS order $\preceq$ on states. But, for simplicity of notation, we omit the ``$c$'', and we write simply: $h$, $\hat{h}$, $\preceq$.

\medskip

We now prove that the bound-based ordering $\preceq$ induced by the canonical good agent map in Def.~\ref{cdf} is a well-quasi-ordering
on the set of all reachable protocol states. We proceed in three steps:
(1) we prove the wqo property on admissible initial states;
(2) we show that every run can be simulated \emph{downwards}, onto a run from a minimal initial state;
(3) we lift the initial-state wqo to all reachable states using run simulation and a componentwise (Dickson) argument on agent counts.

\subsection{WQO on Initial States}

We start by introducing the notion of \emph{minimal initial states}, i.e., they contain no redundant agent instantiations modulo good agent maps. This is formalised in the definition below.

\begin{definition}[Minimal Initial States] \label{def:min-init}
A \emph{minimal initial state} is an admissible initial state that is minimal
with respect to the BOS preorder~$\preceq$, i.e.,  an initial state $s_0\in\Init$
such that there is no distinct $s_0'\in\Init$ with $s_0' \preceq s_0$.
\end{definition}

Now, let us prove that $\preceq$ when restricted to minimal inital states is a w.q.o.

\begin{lemma}[Well-quasi Ordered Initial States]
\label{lem:init-wqo} ~\\
Let $\Init := \bigcup_{n} \Sys(n).\Init$
be the set of all admissible initial states of systems of arbitrary size.
The restriction of $\preceq$ to $\Init$ is a well-quasi-ordering.
\end{lemma}

\begin{proof}
Let
\[
  (s^0_0,s^1_0,s^2_0,\ldots)
\]
be an arbitrary infinite sequence of states in \(Init\). We show that
there exist indices \(p<q\) such that \(s^p_0 \preceq s^q_0\), which by a characterisation in~\cite{FinkelSchnoebelen01} would prove $\preceq$ (over $Init$) to be a w.q.o.

Fix a role \(X\). By the bounded-freshness assumption, every
initial state uses at most \(k_X\) distinct nonce-atoms originating
from agents of role \(X\). Moreover, all admissible initial states are
constructed so that agents of role \(X\) whose indices are congruent
modulo \(k_X\) reuse the same role-originated atoms. Hence, after
applying the canonical map \(h_c\), the local initial components of
all agents of role \(X\) are represented by at most \(k_X\) canonical
components.

Let $c(s_0)$
denote the canonical representative of an initial state \(s_0\), obtained
by applying \(\hat h_c\) to all local components and to the intruder
knowledge, and by retaining only the agents
\[
  1_X,\ldots,k_X
\]
for each role \(X\).

We first observe that the set of possible representatives \(c(s_0)\) is
finite. Indeed, the set of roles is finite; for each role \(X\), there are
only \(k_X\) representative agents; each representative agent is at
control location \(0\) in an initial state; and the set of role-originated
atoms available in representatives is bounded by the freshness bounds
\((k_X)_{X\in Roles}\). Since role descriptions are finite and only
well-typed instantiations are allowed, there are only finitely many
possible initial local components and finitely many possible initial
intruder components, up to the canonical collapse \(\hat h_c\).

Thus the infinite sequence
\[
  c(s^0_0),c(s^1_0),c(s^2_0),\ldots
\]
takes values in a finite set. Therefore, by the pigeonhole principle,
there exist \(p<q\) such that
\[
  c(s^p_0)=c(s^q_0).
\]

We now show that this implies \(s^p_0\preceq s^q_0\). Since
\(c(s^p_0)=c(s^q_0)\), every honest local component of \(s^p_0\)
appears as the \(\hat h_c\)-image of some honest local component of
\(s^q_0\) of the same role. Moreover, the intruder components satisfy
\[
  \hat h_c(K^{s^q_0}_I)=K^{s^p_0}_I .
\]
Hence the two clauses in the definition of the BOS order are satisfied
with witness \(h_c\). Therefore,
\[
  s^p_0\preceq s^q_0 .
\]

Since every infinite sequence in \(Init\) contains such an increasing
pair~\cite{FinkelSchnoebelen01}, the restriction of \(\preceq\) to \(Init\) is a well-quasi-ordering.
\end{proof}

\subsection{Downward run simulation onto minimal initial states.}

We next show that any run of an arbitrarily large system projects, under the
canonical collapse, onto a run of the minimal system obtained by folding its
agents.

\begin{lemma}[Downward run simulation from minimal initial states]
\label{lem:min-init-sim}
For every initial state $s_0 \in \Init$, its canonical collapse
$m_0 := \hatmap{h}(s_0)$ is a $\preceq$-minimal initial state with
$m_0 \preceq s_0$.
Moreover, for every finite run
\[
s_0 \xrightarrow{\pi} s,
\]
there exists a run
\[
m_0 \xrightarrow{\widetilde{h}(\pi)} s'
\qquad\text{with}\qquad
s' = \hatmap{h}(s),
\]
and hence $s' \preceq s$.
\end{lemma}

\begin{proof}
Fix $s_0 \in \Init$ and let $m_0:=\hat h(s_0)$ be its canonical collapse under
$h=h_c$. By construction $m_0$ retains at most $k_X$ agents per role $X$, applying
$h_c$ to it is the identity, so $m_0$ is $\preceq$-minimal, and
$\hat h(s_0)=m_0$ gives $m_0\preceq s_0$ directly from Definition~\ref{bosdefinition}.

\paragraph{Downward simulation of the run.}
We proceed by induction on the length of
$\pi:\; s_0 \xrightarrow{a_1} s_1 \xrightarrow{a_2}\cdots\xrightarrow{a_k} s_k=s$,
showing that the image sequence $\widetilde h(\pi)$ is executable from $m_0$ and
reaches exactly $\hat h(s)$. Recall (paragraph~``Canonical atoms'',
Section~\ref{par:canonical-atoms}) that all reachable states range over the
canonical atoms, so $\hat h$ relabels agent indices without identifying distinct
atoms.

\emph{Base case} ($k=0$): the empty run reaches $m_0=\hat h(s_0)$.

\emph{Inductive step.} Assume the image of $s_0\xrightarrow{\pi'}s_{k-1}$ is a run
$m_0\xrightarrow{\widetilde h(\pi')}\hat h(s_{k-1})$. Consider the last action
$a_k$.

\underline{Honest send $i!\cdot:t$.} Enabledness at $s_{k-1}$ means the sender can
compose $t$ from its local state, i.e.\ $K_i^{s_{k-1}}\vdash_{DY} t$. By
Lemma~\ref{lem:forward-dy}, $\hat h(K_i^{s_{k-1}})\vdash_{DY}\hat h(t)$, i.e.\ the
agent $h(i)$ in $\hat h(s_{k-1})$ can compose $\hat h(t)$, so
$\widetilde h(a_k)=h(i)!\cdot:\hat h(t)$ is enabled. Executing it adds $\hat h(t)$
to the intruder knowledge and advances $h(i)$'s step, which is exactly the image
under $\hat h$ of the update at $s_k$. Hence the reached state is $\hat h(s_k)$.

\underline{Honest receive $i?\cdot:t$.} Enabledness means $K_I^{s_{k-1}}\vdash_{DY} t$.
By Lemma~\ref{lem:forward-dy}, $\hat h(K_I^{s_{k-1}})\vdash_{DY}\hat h(t)$, so the
mapped receive by $h(i)$ of $\hat h(t)$ is enabled from $\hat h(s_{k-1})$, and its
update is the $\hat h$-image of the update at $s_k$; the reached state is
$\hat h(s_k)$.

If $\widetilde h$ merges $a_k$ with the previous action by the stuttering rule
(same agent/step image), the corresponding step is already present and no new
transition is needed; the reached state is unchanged and still equals
$\hat h(s_k)$.

Thus $m_0\xrightarrow{\widetilde h(\pi)}\hat h(s)=:s'$. Finally, taking $s'':=s$
and the same $h$ as witness, $\hat h(s)=s'$ gives $s'\preceq s$ by
Definition~\ref{bosdefinition}.
\end{proof}

~

\noindent
\emph{Note.} In the above, the endpoint $s'$ reached from the ``smaller``,
minimal $m_0$ is the \emph{collapse} $\hatmap{h}(s)$ of the original endpoint
$s$, and it therefore sits \emph{below} $s$ in the order ($s'\preceq s$). 
I.e.,  collapsing agents can  lose distinctions.

A direct consequence of Lemma~\ref{lem:min-init-sim} above is the following lemma:

\begin{lemma}[Reachability from minimal initial states]
\label{reachfrommin}
For every reachable state $s$, there exists a $\preceq$-minimal initial state
$m_0 \in \Init$ and a finite run
\[
m_0 \xrightarrow{\pi'} s'
\]
such that $s' \preceq s$ and $s'=\hatmap{h}(s)$.
\end{lemma}

\begin{proof}
Let $s$ be reachable via $s_0\xrightarrow{\pi}s$ with $s_0\in\Init$. Apply
Lemma~\ref{lem:min-init-sim} to $s_0$ and $\pi$: with $m_0=\hat h(s_0)$ we obtain a
run $m_0\xrightarrow{\widetilde h(\pi)}s'$ with $s'=\hat h(s)\preceq s$.
\end{proof}

\subsection{The BOS is a Well Quasi Order}

And now, we have the final lifting step, showing that our agent-map and induced term-map leads to an w.q.o. over all reachable states:

\begin{theorem}[BOS is a well-quasi-order]
\label{boswqo}
The bound-based ordering $\preceq$ is a well-quasi-ordering on the set of all
reachable states.
\end{theorem}

\begin{proof}
Let $(s_i)_{i\in\mathbb{N}}$ be an infinite sequence of reachable states. We must
find $p<q$ with $s_p\preceq s_q$.

\paragraph{Step 1: a finite invariant.}
By Section~\ref{par:canonical-atoms}, every reachable state uses only the finite
set $A$ of canonical atoms, and honest agents carry role messages of bounded
depth. Hence the canonical collapse $c(s):=\hat h(s)$, which keeps at most $k_X$
representative agents per role and an intruder-knowledge component over $A$,
ranges over a \emph{finite} set of values (finitely many control locations,
finitely many canonical local components, and finitely many knowledge sets over
the bounded-depth terms on $A$). Therefore there is an infinite subsequence
$(s_{i_j})_j$ on which $c(s_{i_j})$ is a single fixed value $c$.

\paragraph{Step 2: Dickson on agent counts.}
Along this subsequence, associate with each $s_{i_j}$ the vector
$v_j:=(\#\{\text{agents of role }X\text{ in }s_{i_j}\})_{X\in\Roles}\in\mathbb{N}^{|\Roles|}$.
By Dickson's lemma\footnote{Dickson's lemma~\cite{Dickson1913}: for any fixed $k$, the set
$\mathbb{N}^{k}$ ordered componentwise is a well-quasi-order (equivalently, every
subset of $\mathbb{N}^{k}$ has finitely many minimal elements, and there is no
infinite antichain). We use the case $k=|\Roles|$; it is also the special case,
for finitely many copies of $(\mathbb{N},\le)$, of the fact that a finite product
of well-quasi-orders is a well-quasi-order~\cite{FinkelSchnoebelen01}.}~\cite{Dickson1913}, $\mathbb{N}^{|\Roles|}$ under the componentwise order is a wqo,
so there are $j<j'$ with $v_j\le v_{j'}$ componentwise.

\paragraph{Step 3: exhibit the BOS witness.}
Set $p:=i_j$, $q:=i_{j'}$. Both $s_p$ and $s_q$ collapse to the same $c$, so they
carry the same intruder knowledge over $A$ and, for each canonical class, the
same canonical local component; and $s_q$ has at least as many agents of each
role as $s_p$. Define a good agent map $g$ that sends the agents of $s_q$ onto
the agents of $s_p$ class by class (surjectively, folding the surplus agents of
$s_q$ onto representatives, which exist since $v_p\le v_q$). Because the atoms are
canonical, $\hat g$ identifies no distinct atoms, so $\hat g(s_q{}_a)$ matches the
corresponding local component of $s_p$ for each honest agent, and
$\hat g(K_I^{s_q})=K_I^{s_p}$. Both clauses of Definition~\ref{bosdefinition} hold,
hence $s_p\preceq s_q$.
\end{proof}

%% file: concreteFig.tex
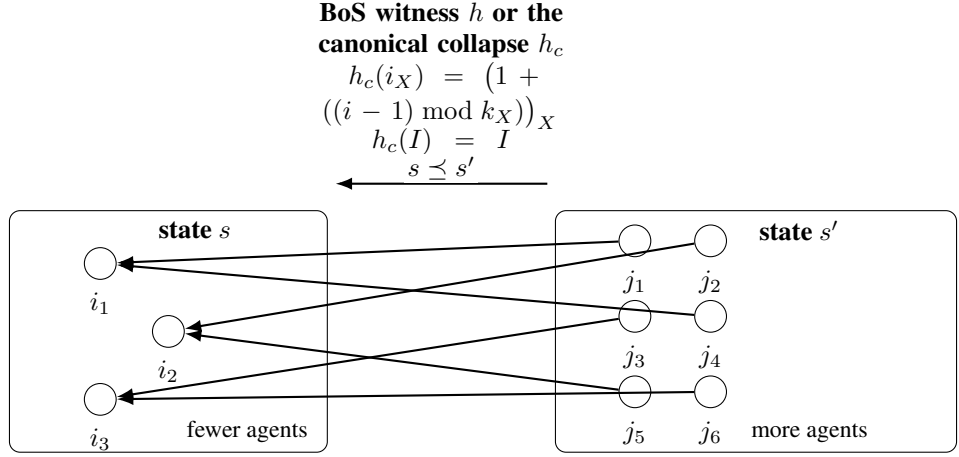
\begin{figure*}[!t]
\centering
\scalebox{1}{
\begin{tikzpicture}[
  font=\small,
  arrow/.style={-{Latex[length=2.2mm]}, thick},
  node distance=10mm and 22mm,
  agent/.style={circle, draw, inner sep=1.2pt, minimum size=4.2mm},
  lbl/.style={inner sep=1pt, fill=white},
  frame/.style={draw, rounded corners, inner sep=6pt}
]

\node[frame, minimum width=4.2cm, minimum height=3.2cm] (s)  {};
\node[frame, minimum width=5.3cm, minimum height=3.2cm, right=3.0cm of s] (sp) {};

\node[anchor=north] at (s.north) {\textbf{~~~~~~~~state $s$}};
\node[anchor=north] at (sp.north) {\textbf{~~~~~~~~~~~~~state $s'$}};

\node[agent] (i1) at ($(s.center)+(-0.9,0.9)$) {};
\node[agent] (i2) at ($(s.center)+(0.0,0.0)$) {};
\node[agent] (i3) at ($(s.center)+(-0.9,-0.9)$) {};

\node[anchor=north] at ($(i1.south)+(0,-1pt)$) {$i_1$};
\node[anchor=north] at ($(i2.south)+(0,-1pt)$) {$i_2$};
\node[anchor=north] at ($(i3.south)+(0,-1pt)$) {$i_3$};

\node[agent] (j1) at ($(sp.center)+(-1.6,1.2)$) {};
\node[agent] (j2) at ($(sp.center)+(-0.6,1.2)$) {};

\node[agent] (j3) at ($(sp.center)+(-1.6,0.2)$) {};
\node[agent] (j4) at ($(sp.center)+(-0.6,0.2)$) {};

\node[agent] (j5) at ($(sp.center)+(-1.6,-0.8)$) {};
\node[agent] (j6) at ($(sp.center)+(-0.6,-0.8)$) {};

\node[anchor=north] at ($(j1.south)+(0,-1pt)$) {$j_1$};
\node[anchor=north] at ($(j2.south)+(0,-1pt)$) {$j_2$};
\node[anchor=north] at ($(j3.south)+(0,-1pt)$) {$j_3$};
\node[anchor=north] at ($(j4.south)+(0,-1pt)$) {$j_4$};
\node[anchor=north] at ($(j5.south)+(0,-1pt)$) {$j_5$};
\node[anchor=north] at ($(j6.south)+(0,-1pt)$) {$j_6$};

\node[anchor=south west] at ($(sp.south west)+(70pt,0pt)$) {\footnotesize more agents};

\node[anchor=south west] at ($(s.south west)+(63pt,0pt)$) {\footnotesize fewer agents};

\draw[arrow] (j1) -- (i1);
\draw[arrow] (j4) -- (i1);
\draw[arrow] (j2) -- (i2);
\draw[arrow] (j5) -- (i2);
\draw[arrow] (j3) -- (i3);
\draw[arrow] (j6) -- (i3);

\node[align=center, text width=4.6cm, anchor=south] (hc)
  at ($(s.north east)!0.5!(sp.north west)+(0,6mm)$) {
  \textbf{BoS witness $h$ or the canonical collapse $h_c$}\\
  $\displaystyle h_c(i_X)=\bigl(1+((i-1)\bmod k_X)\bigr)_X$\\[-1pt]
  $\displaystyle h_c(I)=I$
};

\draw[arrow]   ($(sp.north west)+(-0.1,0.35)$) -- ($(s.north east)+(0.1,0.35)$)
  node[midway, above, lbl] {$s \preceq s'$};

\end{tikzpicture}
}
\caption{BoS via the (canonical) good agent map}
\label{fig:bos-canonical-simple}
\end{figure*}

%% file: wsts_result2.tex
Having established in Section~\ref{sec:wqo} that the BOS relation $\preceq$ is a
well-quasi-order on reachable states and that executions are simulable modulo
$\preceq$, we now build our parameterised-secrecy decidability argument.

The structure of the section is as follows. We first isolate the one direction
of preservation of DY derivability that is valid under a general collapsing map. We then use
this to prove upward compatibility of the transition relation with the BOS
ordering. After that, we identify the finiteness phenomenon induced by bounded
freshness via terminal states and finite atom bases. This makes it possible to
replace the ``raw leakage'' in Section~\ref{sec:leakage}   by a stronger notion of ``bad state'' that is
stable under collapsing, and finally to obtain decidability, in two complementary
ways: a self-contained cut-off argument and a standard WSTS coverability argument.

\subsection{Forward Preservation of Dolev--Yao Derivability} \label{dy-fw}

Before reasoning about states and transitions, we record the basic fact that our
$\hat h$-collapsing preserves DY derivability in the ``forward direction''; this is as follows:

\begin{lemma}[Forward preservation of DY derivability]
\label{lem:forward-dy}
For every $K \subseteq T$ and every $t \in T$,
\[
K \vdash_{DY} t \;\Longrightarrow\; \hat h(K) \vdash_{DY} \hat h(t).
\]
\end{lemma}

\begin{proof}
We argue by induction on a derivation of $K \vdash_{DY} t$ in the DY system of Section~\ref{app:dy}.

\medskip
\noindent \emph{Axiom case.}
If $t \in K$, then $\hat h(t) \in \hat h(K)$, and hence
$\hat h(K) \vdash_{DY} \hat h(t)$ by \textsc{Ax}.

\medskip
\noindent \emph{Pairing case.}
Suppose the last rule is
\[
\frac{K \vdash_{DY} u \qquad K \vdash_{DY} v}{K \vdash_{DY} \langle u,v\rangle}.
\]
By induction hypothesis,
$\hat h(K) \vdash_{DY} \hat h(u)$ and
$\hat h(K) \vdash_{DY} \hat h(v)$.
Since $\hat h$ is homomorphic on pairing,
$\hat h(\langle u,v\rangle)=\langle \hat h(u),\hat h(v)\rangle$,
and therefore
$\hat h(K) \vdash_{DY} \hat h(\langle u,v\rangle)$ by \textsc{Pair}.

\medskip
\noindent \emph{Projection cases.}
Suppose the last rule is \textsc{Fst}; the \textsc{Snd} case is analogous.
By induction hypothesis,
$\hat h(K) \vdash_{DY} \langle \hat h(u),\hat h(v)\rangle$,
hence
$\hat h(K) \vdash_{DY} \hat h(u)$ by \textsc{Fst}.

\medskip
\noindent \emph{Symmetric encryption/decryption cases.}
For \textsc{SEnc}, from $\hat h(K)\vdash_{DY}\hat h(u)$ and $\hat h(K)\vdash_{DY}\hat h(k)$
we get $\hat h(K)\vdash_{DY}\mathsf{enc}_s(\hat h(u),\hat h(k))=\hat h(\mathsf{enc}_s(u,k))$.
For \textsc{SDec}, from $\hat h(K)\vdash_{DY}\mathsf{enc}_s(\hat h(u),\hat h(k))$ and
$\hat h(K)\vdash_{DY}\hat h(k)$ we get $\hat h(K)\vdash_{DY}\hat h(u)$.

\medskip
\noindent \emph{Public-key and hashing cases.}
Identical in structure, using homomorphicity of $\hat h$ on the public-key
constructors and the hash constructor, and the corresponding DY rules.

\medskip
This exhausts the rules of the DY calculus, and the lemma follows.
\end{proof}

\noindent For a
general good map the converse can fail, since $\hat h$ may identify distinct
atoms; for the canonical map committed to in Section~\ref{sec:wqo}, no distinct
atoms are identified (Section~\ref{par:canonical-atoms}), so the two directions
coincide.

\subsection{Upward Compatibility}

The next step is to connect the term-level preservation result above with the
operational semantics. The point of the following theorem is this: \\ ~if a ``smaller''
state (i.e., with fewer, collapsed agents) can perform one step, then any ``larger'' state found above it in the BOS order  can
mimic that step, possibly modulo stuttering collapse.

\begin{theorem}[Upward compatibility]
\label{thm:upward-compat}
If $s \preceq s'$ and $s \xrightarrow{\alpha} t$, then there exists a state
$t'$ such that
\[
s' \xrightarrow{\,\tilde h(\alpha)\,*\,} t'
\qquad\text{and}\qquad
t \preceq t'.
\]
\end{theorem}

\begin{proof}
Assume $s \preceq s'$ witnessed by the canonical good agent map $h$. By
Definition~\ref{bosdefinition}, for every honest agent $i$ in $s$ there is an
honest agent $j$ of the same role in $s'$ with $\hat h(s'_j)=s_i$, and
$\hat h(K_I^{s'}) = K_I^s$.

The key point, from Section~\ref{par:canonical-atoms}, is that $s$ and $s'$ range
over the same canonical atoms and $\hat h$ identifies no distinct atoms. Hence the
BOS witness gives more than a collapse: each honest local component present in
$s$ is \emph{present verbatim} in $s'$ (carried by some agent $j$), and
$K_I^{s'}\supseteq K_I^{s}$ over the common atom set (the larger system has seen
at least the same messages).

Now let $s \xrightarrow{\alpha} t$.

\noindent \emph{Case 1: honest send $\alpha=i!\cdot:u$.}
The agent $i$ is present in $s'$ (as some $j$ at the same local state), so the
same send is enabled at $s'$: it composes $u$ from the identical local knowledge.
Executing it adds $u$ to $K_I^{s'}$ and advances $j$'s step. Since $K_I^{s'}$
already contained $K_I^s$, the resulting intruder knowledge still contains that of
$t$, and all agents of $s$ retain matching components in the successor. Hence
$t\preceq t'$.

\noindent \emph{Case 2: honest receive $\alpha=i?\cdot:u$.}
Enabledness at $s$ means $K_I^{s}\vdash_{DY} u$. Since $K_I^{s}\subseteq K_I^{s'}$
and derivability is monotone, $K_I^{s'}\vdash_{DY} u$, so the same receive by $j$
of $u$ is enabled at $s'$. Its update matches that at $t$, so $t\preceq t'$.

In both cases, no reverse form of Lemma~\ref{lem:forward-dy} is invoked: the
larger state simulates the smaller one because it \emph{contains} the relevant
agents and knowledge over the common canonical atoms. If $\tilde h(\alpha)$
stutters, the argument applies step by step along the finite image path.
\end{proof}

~

The result above is one of the usual
 properties, linked to monotonicity, needed for a system to be a WSTS~\cite{FinkelSchnoebelen01}; we continue on our way to proving that in the next subsections.

\subsection{Terminality and Finite Atom Bases}

We now isolate the finiteness phenomenon induced by bounded freshness. The role
of terminal states is that, once no new atoms can be created, all later DY
reasoning takes place over a fixed finite set of atoms. This is exactly what
will allow us to normalise leaks into a stable form.

\begin{definition}[Terminal state]
A state is terminal if no transition can introduce fresh atoms.
\end{definition}

\begin{lemma}[Terminality]
\label{lem:terminality}
Every run reaches a terminal state.
\end{lemma}

\begin{proof}
For each role $X$, the bounded-freshness assumption fixes a
budget $k_X$ on the number of fresh atoms that can be introduced by that role
along a run. Summing over all roles, the total number of freshness-producing
actions in any run is bounded by $B=\sum_X k_X$.
Once these $B$ opportunities have been exhausted, no transition can introduce a
fresh atom. Therefore, every run has a suffix whose first state is terminal.
\end{proof}

\begin{lemma}[Finite atom basis]
\label{lem:finite-atoms}
After terminality, only finitely many atoms occur in the system.
\end{lemma}

\begin{proof}
By bounded freshness, only finitely many fresh atoms can be
introduced before the first terminal state (Lemma~\ref{lem:terminality}).
Moreover, the initial state contains only finitely many atoms.
Hence, the set of atoms that can occur at or after a terminal
state is the union of two finite sets, and is therefore finite.
\end{proof}

\subsection{Bad States for Parameterised Secrecy }

We first introduce the  notion of ``bad state'' directly encoding  a state of leakage  in the Definition~\ref{defpsp} of the parameterised secrecy problem (PSP), i.e., a state where the
intruder can derive a secret:

\begin{definition}[Bad state] \label{bad}
A state $s$ is \emph{bad} if the intruder can derive some secret atom from its
current knowledge:
\[
s \text{ is bad } \quad\text{iff}\quad \exists a \in Secret.\; K_I^s \vdash_{DY} a.
\]
We also write
\[
Bad := \{\, s \mid \exists a \in Secret.\; K_I^s \vdash_{DY} a \,\}.
\]
\end{definition}

For a \emph{general} good map, the set $Bad$
is not guaranteed to be upward-closed under $\preceq$.
We therefore define a strengthened notion of bad states:

\begin{definition}[Stable bad states]
A state \(s\) is \emph{stably bad} if there exists a secret
\(a\in Secret\) such that:
\begin{enumerate}
\item \(K^s_I \vdash_{\mathrm{DY}} a\);
\item there exists a DY derivation \(D\) of \(a\) from \(K^s_I\)
such that \(D\) is liftable along BOS witnesses: i.e., \\
\hspace*{0.2cm} for every
state \(s'\) with \(s\preceq s'\), witnessed by \(h\), the same
\hspace*{0.2cm} derivation schema can be reconstructed in \(K^{s'}_I\), yielding
\hspace*{0.2cm} some \(a'\in Secret\) with
\[
  \hat h(a') = a
  \qquad\text{and}\qquad
  K^{s'}_I \vdash_{\mathrm{DY}} a' .
\]
\end{enumerate}

\hspace*{0.2cm}
We write
\[
  Bad^\star := \{\,s \mid s \text{ is stably bad}\,\}.
\]
\end{definition}

\paragraph*{\underline{Explanation}}
A stable bad state is one where the leak is witnessed by a DY proof
that does not rely on identifications  created by collapsing, i.e., the leak depends only on explicitly present terms, not on equalities that could be
 introduced by collapsing. Equivalently,
the proof of the leak can be lifted to any ``larger'' state $s'$ that collapses
to $s$. (For the canonical map of Section~\ref{sec:wqo} no such identifications
arise, so every bad state is already stably bad.)

~

The aim here is to prove that any leak (i.e., done via a $Bad$ state) can be normalised into a leak via a $Bad^\star$ state; this would help us prove that our systems are WSTSs. We proceed step by step in that direction.

\begin{lemma}[Soundness]
\label{lem:sound}
If $s \in Bad^\star$, then $s \in Bad$.
\end{lemma}

\begin{proof}
Immediate from the definition.
\end{proof}

\begin{lemma}[Upward closure]
\label{lem:bad-upward}
If $s \in Bad^\star$ and $s \preceq s'$, then $s' \in Bad^\star$.
\end{lemma}

\begin{proof}
Let \(s\in Bad^\star\), so $K^s_I\vdash_{DY} a$ for some $a\in Secret$, and
\(s\preceq s'\) witnessed by the canonical $h$. By Section~\ref{par:canonical-atoms}
the map $\hat h$ identifies no distinct atoms, so $K^s_I\subseteq K^{s'}_I$ over the
common canonical atoms. By monotonicity of $\vdash_{DY}$, $K^{s'}_I\vdash_{DY} a$
with the \emph{same} secret $a$ and the same derivation, which is therefore
liftable. Hence \(s'\in Bad^\star\).
\end{proof}

\subsection{Normalisation of Leaks}

We next show that the  notion of stable-badness is not weaker than the badness itself, in the presence of bounded freshness. The main point is that any genuine leak can
be witnessed, after terminality, by a derivation over a finite atom basis, and
such a derivation can be normalised so that it uses only atoms already present
in the current intruder knowledge, i.e., we can characterise any leaks via stably bad states (not just bad states).

\begin{lemma}[Normalisation]
\label{lem:normalisation}
Every leaky run contains a stably bad state.
\end{lemma}

\begin{proof}
Let $\pi = s_0 \to \cdots \to s_j$ be a minimal leaky run, i.e.,
$K_I^{s_j} \vdash_{DY} a$ for some $a \in Secret$, and for all $i < j$,
$K_I^{s_i} \not\vdash_{DY} a$.
Let $\tau$ be the index of the first terminal state (Lemma~\ref{lem:terminality}).

\medskip
\noindent\emph{Case 1: $j \le \tau$.}
Since $s_j$ is the first leaky state, any derivation of $a$ from $K_I^{s_j}$
cannot rely on identifications between distinct atoms not already present in
$K_I^{s_j}$; and, over canonical atoms, the canonical collapse creates no such
identifications. Hence $s_j \in Bad^\star$.

\medskip
\noindent\emph{Case 2: $j > \tau$.}
After $s_\tau$, no new atoms are introduced. By Lemma~\ref{lem:finite-atoms}, all
atoms occurring in $K_I^{s_j}$ belong to a fixed finite set $A$. By the normal-form
theorem for DY derivations (Lemma~\ref{cfp}), there is a cut-free derivation $D$
of $a$ from $K_I^{s_j}$; choose one of minimal height. Every inference of $D$ then
uses only atoms occurring in $K_I^{s_j}$ and depends on no identification of
distinct atoms, so $D$ is liftable along BOS witnesses and $s_j\in Bad^\star$.
\end{proof}

~

Next, we actually show that if there is a leak, we can characterise it via a stably bad state that we can find on a corresponding leaky run:

\begin{proposition}[Bounded leak witness]
\label{prop-bound}
Let $Pr$ be a protocol and let $(k_X)_{X \in Roles}$ be the
freshness bounds. Then, there exists a computable function
\[
f = f(|Pr|, \#bound)
\]
such that for every leaky run
\[
\pi = s_0 \to s_1 \to \cdots
\]
there exists an index $j \leq f(|Pr|, \#bound)$ such that
$s_j \in Bad^\star$.
\end{proposition}

\begin{proof}
Let $\pi = s_0 \to \cdots \to s_j$ be a minimal leaky run.
By Lemma~\ref{lem:terminality}, let $s_\tau$ be the first terminal state.
If $j \le \tau$, the result follows immediately (and $\tau\le\sum_X k_X$).

Assume $j > \tau$. By Lemma~\ref{lem:finite-atoms}, all atoms occurring
after $s_\tau$ belong to a finite set $A$. After $s_\tau$ no new atoms are
introduced and all intruder derivations are performed over the fixed finite basis
$A$ using only well-typed substitutions. Hence any derivation of a secret from
$K_I^{s_j}$ admits a cut-free normal form (Lemma~\ref{cfp}) of height bounded by a
function $g(|A|)$. Cut-elimination introduces no new atoms, and well-typedness
prevents type-flaw encodings, so the bound $g(|A|)$ applies directly.

Each inference step corresponds to at most a constant number of protocol or DY
transitions, so the leak occurs within $\tau + c \cdot g(|A|)$ steps. Since $|A|$
is bounded by the freshness bounds and the protocol size,
$f(|Pr|, \#bound) = \tau + c \cdot g(|A|)$ is computable. By
Lemma~\ref{lem:normalisation}, the corresponding state belongs to $Bad^\star$.
\end{proof}

\subsection{WSTS Structure}

We can now package the preceding ingredients into a final argument that our protocol-models are WSTSs. The
ordering $\preceq$ provides the well-quasi-order, Theorem~\ref{thm:upward-compat}
provides monotonicity of the transition system, and
Lemma~\ref{lem:bad-upward} provides upward closure of the target set. So, we have the following theorem:

\begin{theorem} \label{th:wsts}
$(S,\rightarrow,\preceq)$ is a well-structured transition system with respect
to $Bad^\star$.
\end{theorem}

\begin{proof}
We have this from the results above:
\begin{itemize}
\item $\preceq$ is a wqo (Section~\ref{sec:wqo});
\item upward compatibility holds (Theorem~\ref{thm:upward-compat});
\item $Bad^\star$ is upward-closed (Lemma~\ref{lem:bad-upward}).
\end{itemize}
\end{proof}

\subsection{Decidability}

We now spell out the final decidability results. We give two routes. The
\emph{cut-off} route is self-contained: it rests only on the downward simulation
(Lemma~\ref{lem:min-init-sim}), forward preservation of derivability
(Lemma~\ref{lem:forward-dy}), and the bounded-witness Proposition~\ref{prop-bound};
it does not use upward compatibility. The \emph{WSTS-coverability} route reduces
parameterised secrecy to coverability of the upward-closed set $Bad^\star$ and
runs a standard backward reachability from $Bad^\star$. Both are captured by
Theorem~\ref{th:deci}, whose proof gives both arguments.

\begin{theorem}[Decidability] \label{th:deci}
Parameterised secrecy is decidable.
\end{theorem}

\begin{proof}
We give two arguments; either suffices.

\paragraph{Cut-off argument (self-contained).}
By Theorem~\ref{thm:cutoff}, parameterised secrecy fails iff some system $S(m)$
with $m\le c=(k_X)_X$ reaches a bad state. There are finitely many size vectors
$m\le c$. For each, $S(m)$ has a fixed, finite number of agents; by
Lemma~\ref{lem:terminality} and Lemma~\ref{lem:finite-atoms} every run reaches a
terminal state over a finite atom basis, and by Proposition~\ref{prop-bound} any
leak in $S(m)$ appears within $f(|Pr|,\#bound)$ steps. Enumerating runs up to that
bound (modulo $\preceq$-subsumption on the finitely many minimal initial states of
$S(m)$) is a terminating search that decides whether $S(m)$ reaches a bad state.

\paragraph{WSTS coverability argument.}
Alternatively, by Theorem~\ref{th:wsts}, $(S,\rightarrow,\preceq)$ is a WSTS and
$Bad^\star$ is upward-closed; hence coverability of $Bad^\star$ from $Init$ is
decidable by backward reachability, and by Lemma~\ref{lem:sound} and
Lemma~\ref{lem:normalisation} it holds iff parameterised secrecy fails.
\end{proof}

%% file: decid-alg2.tex
We now turn the structural results in the previous two sections into concrete decision
procedures for the parameterised secrecy problem.

The key ingredients established so far are:

\begin{itemize}
\item the bound-based ordering (BOS) $\preceq$ is a well-quasi-order on reachable
states (Theorem~\ref{boswqo});

\item the downward run simulation collapses any run onto a run from a minimal
initial state (Lemma~\ref{lem:min-init-sim});

\item forward preservation of DY derivability (Lemma~\ref{lem:forward-dy}), which
takes badness downward under collapse; and, for the coverability route,
$\preceq$ is upward compatible (Theorem~\ref{thm:upward-compat}) and
$Bad^\star$ is upward-closed (Lemma~\ref{lem:bad-upward});

\item every leaky run contains a stably bad state within a computable number of
steps (Lemma~\ref{lem:normalisation}, Proposition~\ref{prop-bound}).
\end{itemize}

The  points above, together, lead to our main result, Theorem~\ref{thm:cutoff} below.

\subsection{A  Cutoff for Parameterised Secrecy}

\begin{definition}[{BOS cutoff $c$}]
For each role $X$,  let  $k_X$ be its associated freshness bound.
The vector $c = (k_X)_{X \in Roles}$ is  called a \emph{BOS cutoff}.
\end{definition}

The following theorem formalises the intuition that secrecy violations
in arbitrarily large systems can already be witnessed in systems of bounded
size, i.e., unbounded systems, in our model, exhibit no fundamentally new
secrecy-violating behaviours beyond the BOS cutoff:

\begin{theorem}[DY-sound cutoff for parameterised secrecy]
\label{thm:cutoff}
For the parameterised protocol system with DY intruder and BOS order $\preceq$,
the following are equivalent:
\begin{enumerate}
\item There exists a leaky run in some system $S(n)$;
\item There exists a run reaching a bad state in some system  $S(m)$ of size
$m \leq c$ (componentwise), starting from a minimal initial state.
\end{enumerate}
\end{theorem}

For the canonical map, a bad state is already stably bad, so the decision
procedures below may test membership in $Bad$ directly.

\begin{proof}
$(2 \Rightarrow 1)$ is immediate: a system of size $m\le c$ is an instance of the
family, and any bad state is a genuine leak (Definition~\ref{bad}).

$(1 \Rightarrow 2)$. Suppose $S(n)$ has a leaky run $s_0\xrightarrow{\pi}s$ with
$s$ bad, i.e.\ $K_I^{s}\vdash_{DY} a$ for some $a\in Secret$. Apply the downward
simulation of Lemma~\ref{lem:min-init-sim}: with $m_0=\hat h(s_0)$ a minimal
initial state, there is a run $m_0\xrightarrow{\widetilde h(\pi)}\hat h(s)$.

We claim $\hat h(s)$ is bad. Indeed $K_I^{\hat h(s)}=\hat h(K_I^{s})$ and, by
forward preservation (Lemma~\ref{lem:forward-dy}), $K_I^{s}\vdash_{DY} a$ implies
$\hat h(K_I^{s})\vdash_{DY}\hat h(a)$. As $h$ is role-respecting on
agent-indexed secrets, $\hat h(a)\in Secret$. Hence $\hat h(s)$ is bad.

Finally, $m_0$ retains at most $k_X$ agents per role $X$ (Definition~\ref{def:min-init}
and Definition~\ref{cdf}), so $\hat h(s)$ lives in a system of size $m\le c$
componentwise. This exhibits the required run.
\end{proof}

\medskip
\noindent
\emph{Remark on the direction of the cut-off.} The witness is transferred by
collapsing the large run \emph{down} to the minimal system, and badness is
preserved \emph{downward} by forward DY-derivability. This is the direction the
maps support; it does not use upward closure of the bad set. Upward closure
(Lemma~\ref{lem:bad-upward}) and upward compatibility
(Theorem~\ref{thm:upward-compat}) are what the alternative WSTS-coverability
procedure (Algorithm~2) rests on instead.

\subsection{Algorithm 1: Finite Cutoff Exploration}

By Theorem~\ref{thm:cutoff}, it suffices to search, for each of the finitely many
size vectors $m\le c$, whether the fixed-size system $S(m)$ reaches a bad state.
For a minimal initial state $s_0$, the \emph{finite reachability tree}
$\mathsf{FRT}(s_0)$~\cite{FinkelSchnoebelen01} explores all successors while
quotienting modulo $\preceq$ (a new state is not expanded if it is already
$\preceq$-covered). By the wqo property $\mathsf{FRT}(s_0)$ is finite, and by
Proposition~\ref{prop-bound} any leak appears within a computable depth
$f(|Pr|,\#\text{bound})$. Algorithm~\ref{alg:cutoff} explores all such trees up to
the cutoff.

\begin{algorithm}[H]
\caption{Finite cutoff exploration} \label{alg:cutoff}
\begin{algorithmic}[1]
\Require Protocol $Pr$, bounds $(k_X)$, BOS order $\preceq$
\Ensure secure if no leaky run exists, else insecure
\State Compute $c = (k_X)_{X\in Roles}$
\ForAll{role-size vectors $m \leq c$}
  \ForAll{minimal initial states $s_0$ of size $m$}
    \State Construct $FRT(s_0)$ under $\preceq$-subsumption
    \If{some node $s$ satisfies $s \in Bad$}
        \State \Return insecure
    \EndIf
  \EndFor
\EndFor
\State \Return secure
\end{algorithmic}
\end{algorithm}

\noindent\emph{Correctness.}
If Algorithm~\ref{alg:cutoff} returns \emph{insecure}, a bad state has been
reached, hence a genuine secrecy violation exists. If it returns \emph{secure},
no bad state is reachable in any system of size $\le c$; by
Theorem~\ref{thm:cutoff}, no leaky run exists at all.

\subsection{Algorithm 2: WSTS Backward Coverability}

Algorithm~2 avoids an explicit cutoff.  Instead,
having shown $(S,\rightarrow,\preceq)$ to be a WSTS with upward-closed target
$Bad^\star$ (Theorem~\ref{th:wsts}), one runs backward coverability from a finite
basis of $Bad^\star$.

\begin{algorithm}[H]
\caption{WSTS backward coverability}
\begin{algorithmic}[1]
\Require Protocol $Pr$, BOS order $\preceq$
\Ensure secure if initial states are disjoint from backward closure
\State Construct a finite basis $B$ of $Bad^\star$
\State Initialise worklist $W := B$
\While{$W \neq \emptyset$}
  \State remove some $s$ from $W$
  \ForAll{predecessors $p$ of $s$}
    \If{$p$ is not $\preceq$-covered}
      \State add $p$ to $W$
    \EndIf
  \EndFor
  \If{$p$ is a minimal initial state (up to $\preceq$)}
    \State \Return insecure
  \EndIf
\EndWhile
\State \Return secure
\end{algorithmic}
\end{algorithm}

\noindent\emph{Correctness.}
Termination follows from standard WSTS arguments: $\preceq$ is a well-quasi-order
and predecessor generation is effective. If the algorithm returns
\emph{insecure}, a minimal initial state reaches $Bad^\star$, hence a genuine
secrecy violation exists. If a secrecy violation exists, then by
Lemma~\ref{lem:normalisation} there is a state in $Bad^\star$ reachable from some
initial state; since $Bad^\star$ is upward-closed and $(S,\to,\preceq)$ is a WSTS,
backward coverability discovers such a predecessor.

\subsection{Complexity} 

\paragraph{\textbf{Parameters}}
Let $|Pr|$ denote the size of the protocol description (total number of role steps
and atomic symbols in role specifications). Let $(k_X)_{X \in Roles}$ be the
freshness bounds, and write $K := \sum_{X \in Roles} k_X$. Let $c = (k_X)_{X \in Roles}$
be the BOS cutoff. We write $A$ for the set of atoms that may occur after
terminality; by Lemma~\ref{reachfrommin}, $|A| = \mathcal{O}(K + |Pr|)$.

\subsubsection{\textbf{Algorithm 1}}
For a fixed size $m \leq c$, the number of agents is bounded by $K$, each agent has
finitely many control locations (bounded by $|Pr|$), and intruder knowledge
consists of terms over the finite atom set $A$. Modulo $\preceq$-subsumption, the
number of distinct states in the reachability tree is bounded by
$2^{poly(|A|)} = 2^{poly(K + |Pr|)}$, and transition branching is polynomial in $K$
and $|Pr|$. Hence Algorithm~1 runs in time $2^{poly(K + |Pr|)}$. If the bounds
$k_X$ are themselves exponential in $|Pr|$ (as assumed in Section~\ref{sec:wqo}),
the overall complexity is at most $2^{2^{poly(|Pr|)}}$, i.e., double-exponential.

\subsubsection{\textbf{Algorithm 2}}
Backward coverability in WSTS is, in general, nonelementary in the size of the
system description. In our setting, derivation height is bounded by the normal-form
function $g(|A|)$ (Section~\ref{app:dy}) and $|A| = \mathcal{O}(K + |Pr|)$, so
there is a computable $F$, growing faster than any fixed tower of exponentials,
with Algorithm~2 running in time $F(K + |Pr|)$.

%% file: scope-limitations.tex

We now discuss the limits and reach of our results.

\paragraph{\textbf{Primitives and the collapse}}
Our protocol-class provides pairing, hashing, and symmetric and public-key encryption
with atomic keys. The forward-preservation argument
(Lemma~\ref{lem:forward-dy}) is a case analysis over exactly these
constructors, so a further \emph{atomic-keyed} primitive can be added by
extending that analysis with its synthesis and analysis rules. What does
\emph{not} lift for free is a primitive with an equational theory, or
\emph{compound keys}, that is, keys built from several terms (as with a key
derived by a key-derivation function from a secret and context labels). Compound
keys break two assumptions at once: the induced term map is defined by replacing
agent-indexed \emph{atoms}, and the analysis rules for decryption test for an
atomic key in the knowledge. Lifting the result to compound keys or to a general
equational theory is therefore not a matter of adding cases; it requires a
collapse that is stable under the equational theory, which we leave open and
discuss as the main avenue for extension.

\paragraph{\textbf{The well-typed attacker, and how it is realised}}
We restrict the Dolev--Yao attacker to well-typed substitutions. This excludes
type-flaw attacks and is, with bounded freshness, what makes the collapse sound.
It is a clear restriction from classical  Dolev--Yao. It is, however, an
\emph{implementable in practices} via a
standard typing discipline, in which every message field carries an authenticated  header for typing, forcing all attacker  to respect types. 

\paragraph{\textbf{Properties, roles, and state}}
Our target is secrecy, a reachability property. Trace-based properties such as
agreement are natural next targets (they too are monotone under collapse), but
 privacy-style equivalences are not reachability
properties and would need a different ordering. We assume each agent instantiates
a single role with indices unused in message content; an agent playing several
roles can be encoded by a product role, enlarging the per-role bounds. Stateful protocols, with a mutable global store persisting across sessions, fall
outside the model, since the collapse is defined on per-agent local states and
the intruder knowledge, not on shared mutable state. Finally, bounded freshness is
an assumption on the model: where material is drawn from a bounded, rotated, or
cached pool it is met and the cut-off applies. 
Instead, the classical undecidability of the
unbounded-freshness, unbounded-session regime (see Figure~\ref{fig:landscape})
remains.

%% file: concl.tex
This paper addresses secrecy verification for cryptographic protocols in
the Dolev--Yao model from a genuinely parameterised perspective. Unlike
classical approaches, which obtain decidability by fixing the number of
sessions or agents, we consider systems with an unbounded number of role
instantiations and ask whether secrecy can be decided uniformly for all
system sizes. Our main result shows that parameterised secrecy is decidable under
bounded freshness per role, and well-typed substitutions. Technically, protocol executions are modelled
as infinite-state transition systems equipped with a bound-based state
ordering, which we show yields a cut-off and, more strongly, a well-structured transition system.

Methodologically, the paper shows that parameterised verification can be
combined with symbolic protocol analysis without weakening the standard
Dolev--Yao intruder model. The agent-collapsing maps and induced term
abstractions isolate why intruder derivability is preserved under parameter
reduction, and may be of independent interest.

Conceptually, the result clarifies the role of freshness in secrecy
verification. While classical work established that bounding the number
of nonces suffices for decidability, we show that
bounded freshness also controls the parameterised dimension: even with
unboundedly many participants, secrecy violations admit finite cut-offs
and are always witnessed in systems of bounded size.


 Several directions for future work remain open. Extending the framework
 beyond secrecy to properties such as authentication and privacy
 assertions is a natural next step. It would also be interesting to
 relax freshness bounds or replace them with alternative notions of
 boundedness, such as bounded message length. Finally, an important
 practical challenge is to exploit the cut-off results established here
 in the design of automated verification tools for unbounded-session
 protocols.